\documentclass[a4paper,fleqn,11pt]{article}

\usepackage{amsmath}
\usepackage{amsthm}
\usepackage{amssymb}
\usepackage{accents}

\usepackage[a4paper,top=3cm, bottom=3cm, left=3cm, right=3cm]{geometry}
\usepackage[shortlabels,inline]{enumitem}
\usepackage[pdftex]{graphicx}
\usepackage{subcaption}
\usepackage[pdftex,ocgcolorlinks,pagebackref=false]{hyperref}
\usepackage{authblk}

\setlist[enumerate,1]{label={(\roman*)}}

\theoremstyle{plain}
\newtheorem{theorem}{Theorem}[section]
\newtheorem{lemma}[theorem]{Lemma}
\newtheorem{proposition}[theorem]{Proposition}
\newtheorem{corollary}[theorem]{Corollary}

\theoremstyle{definition}
\newtheorem{definition}[theorem]{Definition}

\newcommand{\norm}[2][]{\left\|#2\right\|_{#1}}
\newcommand{\ball}[2]{B_{#1}(#2)}
\newcommand{\ket}[1]{\left|#1\right\rangle}

\newcommand{\ketbra}[2]{\left|#1\middle\rangle\!\middle\langle#2\right|}

\newcommand{\setbuild}[2]{\left\{#1\middle|#2\right\}}
\DeclareMathOperator{\boundeds}{\mathcal{B}}
\DeclareMathOperator{\states}{\mathcal{S}}
\DeclareMathOperator{\Tr}{Tr}
\DeclareMathOperator{\id}{id}
\newcommand{\entropy}{H}
\newcommand{\relativeentropy}[3][]{\mathop{D_{#1}}\mathopen{}\left(#2\middle\|#3\right)\mathclose{}}
\newcommand{\distributions}[1][]{\mathcal{P}_{#1}}
\newcommand{\unittensor}[1]{\langle{#1}\rangle}

\newcommand{\partitions}[1][]{\mathcal{P}_{#1}}

\newcommand{\ed}{\mathop{}\!\mathrm{d}}

\newcommand{\reals}{\mathbb{R}}
\newcommand{\complexes}{\mathbb{C}}
\newcommand{\naturals}{\mathbb{N}}
\newcommand{\nonnegativereals}{\mathbb{R}_{\ge 0}}

\newcommand{\ratefunction}[2]{I_{#2}(#1)}

\newcommand{\Kron}[1][]{\textnormal{Kron}_{#1}}
\newcommand{\LR}[1]{\textnormal{LR}_{#1}}
\newcommand{\logupperLOCC}[1]{E^{#1}}
\newcommand{\loglowerLOCC}[1]{E_{#1}}
\newcommand{\logasymptoticlowerLOCC}[1]{\undertilde{E}_{#1}}
\newcommand{\upperLOCC}[1]{F^{#1}}
\newcommand{\lowerLOCC}[1]{F_{#1}}
\newcommand{\asymptoticlowerLOCC}[1]{\undertilde{F}_{#1}}
\DeclareMathOperator{\Res}{Res}

\title{A family of multipartite entanglement measures}
\author[1,2]{P\'eter Vrana}
\affil[1]{Institute of Mathematics, Budapest University of Technology and Economics, Egry J\'ozsef u. 1., 1111 Budapest, Hungary}
\affil[2]{MTA-BME Lend\"ulet Quantum Information Theory Research Group}
\date{\today}

\begin{document}
\maketitle

\begin{abstract}
We construct a family of additive entanglement measures for pure multipartite states. The family is parametrised by a simplex and interpolates between the R\'enyi entropies of the one-particle reduced states and the recently-found universal spectral points (Christandl, Vrana, and Zuiddam, STOC 2018) that serve as monotones for tensor degeneration.
\end{abstract}

\section{Introduction}

The two main approaches to quantifying entanglement (and more generally, to information quantities) are the operational and the axiomatic one. Operational entanglement measures aim at directly characterising the performance in quantum information processing protocols such as entanglement assisted quantum communication or secret key distillation. In contrast, the axiomatic approach starts with a list of properties (see e.g. \cite{plenio2007introduction} or \cite[Table 3.2]{christandl2006structure} for examples), desirable from a mathematical point of view or believed to grasp some physical aspect of nonlocality, and then seeks for quantities satisfying them. However, such a list of requirements tends to be subjective and also reflects the intended range of applications (e.g. some properties are only relevant in an asymptotic context) and it has been argued that the only requirement for an entanglement measure should be monotonicity under local operations and classical communication (LOCC) \cite{vidal2000entanglement}. In addition, constructing entanglement measures that satisfy a given set of axioms is often a challenging task. For example, additivity is a valuable property in asymptotic settings, but there are only a handful of additive entanglement measures known, especially in the multipartite setting \cite{vidal2002computable,christandl2004squashed,yang2008additive,yang2009squashed}. 

In this paper we construct a family of additive entanglement monotones that are defined on multipartite pure states. The construction makes use of the rate function for a simultaneous version of the spectrum estimation procedure proposed by Alicki, Rudicki and Sadowski \cite{alicki1988symmetry} and by Keyl and Werner in \cite{keyl2001estimating} (see Sections~\ref{sec:notations} and~\ref{sec:estimation} for details). Let the set of decreasingly ordered spectra (with multiplicities) be $\overline{\partitions}$. When applying their measurement to the marginals of a $k$-partite state $\ket{\psi}$, we obtain a sequence of probability distributions on the space of $k$-tuples of ordered spectra (i.e. the possible outcomes of the measurements). This sequence satisfies a large deviation principle with a rate function whose value at $\overline{\lambda}\in\overline{\partitions}^k$ will be denoted by $\ratefunction{\overline{\lambda}}{\psi}$ (this is a special case of the $\mu$-capacity from \cite{franks2020minimal} and the rate function from \cite{botero2020large}, when the compact group is a product of unitary groups). It will also be convenient to introduce a notation for the weighted average of a collection of entropies. For a probability distribution $\theta$ on the set $[k]=\{1,2,\ldots,k\}$ we set
\begin{equation}
\entropy_{\theta}(\overline{\lambda})=\sum_{j=1}^k\theta(j)\entropy(\overline{\lambda_j}).
\end{equation}
We consider the quantity
\begin{equation}\label{eq:multipartitemonotones}
\logupperLOCC{\alpha,\theta}(\ket{\psi}) = \sup_{\overline{\lambda}\in\overline{\partitions}^k}\left[(1-\alpha)\entropy_{\theta}(\overline{\lambda})-\alpha \ratefunction{\overline{\lambda}}{\psi}\right]
\end{equation}
where $\alpha\in(0,1]$. For $\alpha=0$ we extend by continuity: in this limit the value of the rate function is not relevant anymore, but the supremum is restricted to the set $\setbuild{\overline{\lambda}\in\overline{\partitions}^k}{\ratefunction{\overline{\lambda}}{\psi}\neq 0}$, which is the entanglement polytope \cite{walter2013entanglement} of the state $\ket{\psi}$. Therefore the $\alpha\to 0$ limit reduces to the quantum functionals of \cite{christandl2018universal}. We show that $\logupperLOCC{\alpha,\theta}(\ket{\psi}\otimes\ket{\varphi})=\logupperLOCC{\alpha,\theta}(\ket{\psi})+\logupperLOCC{\alpha,\theta}(\ket{\varphi})$ and that if there is an LOCC channel mapping $\ket{\psi}$ to $\ket{\varphi}$ then $\logupperLOCC{\alpha,\theta}(\ket{\psi})\ge\logupperLOCC{\alpha,\theta}(\ket{\varphi})$.

In fact, we prove the stronger statement that $\upperLOCC{\alpha,\theta}(\ket{\psi})=2^{\logupperLOCC{\alpha,\theta}(\ket{\psi})}$ is an element of the asymptotic spectrum of LOCC transformations as introduced in \cite{jensen2019asymptotic}. As the proof relies on the theory developed there, we provide here a brief overview and explain the connection to \cite{christandl2018universal} and the asymptotic restriction of tensors.

Asymptotic tensor restriction can be viewed as a weak notion of asymptotic entanglement transformation of pure states \cite{chitambar2008tripartite}, where we require the target to be reached exactly with an arbitrarily low but nonzero probability of success, as proposed in \cite{bennett2000exact} for single copies. In \cite{strassen1988asymptotic} Strassen proved a powerful characterisation of asymptotic restriction, which can be reformulated as the following dual characterisation of the optimal transformation rate. If $\psi$ and $\varphi$ are tensors of order $k$ (which we may interpret as state vectors of $k$-partite states) and $a\ge b$ means that $a$ restricts to $b$ (i.e. can be converted via SLOCC \cite{dur2000three}), then
\begin{equation}\label{eq:tensorrate}
\sup\setbuild{R\in\nonnegativereals}{\psi^{\otimes n}\ge\varphi^{\otimes Rn+o(n)}\text{ for large $n$}}=\inf_{f\in\Delta(\mathcal{T}_k)}\frac{\log f(\psi)}{\log f(\varphi)},
\end{equation}
where $\Delta(\mathcal{T}_k)$ is the asymptotic spectrum of tensors, defined as the set of real valued functions on (equivalence classes of) tensors that are
\begin{enumerate}[(S1)]
\item\label{it:tensormonotone} monotone under restriction
\item\label{it:tensormultiplicative} multiplicative under the tensor product
\item\label{it:tensoradditive} additive under the direct sum
\item\label{it:tensornormalised} normalised to $r$ on the unit tensor of rank $r$ (unnormalised $r$-level GHZ state $\ket{11\ldots 1}+\ket{22\ldots 2}+\cdots+\ket{rr\ldots r}$, denoted $\unittensor{r}$).
\end{enumerate}

Tensor restrictions or stochastic entanglement transformations provide no control on the probabilities and do not offer any meaningful notion of approximate transformations. As a refinement of the problem, one may investigate the trade-off between the transformation rate and the exponential rate at which the error approaches zero or one on either side of the optimal rate. This has been answered in \cite{hayashi2002error} for bipartite entanglement concentration, allowing either an error (i.e. nonzero distance from the target state) or a probability of failure (i.e. nonzero chance of not reaching the target state). They find that, as with many information processing tasks, these trade-off relations can be characterised in terms of R\'enyi information quantities, in this case R\'enyi entanglement entropies.

The main result of \cite{jensen2019asymptotic} provides a characterisation like \eqref{eq:tensorrate} of the error exponents for probabilistic entanglement transformations of pure multipartite states in the converse regime. This means that for each large $n$ we require an LOCC protocol that, when run on $\psi^{\otimes n}$ as the input state, declares success with probability at least $2^{-rn}$, in which case the resulting state is exactly $\varphi^{\otimes Rn+o(n)}$. For a given $r>0$ the maximal achievable $R$ is
\begin{equation}\label{eq:LOCCrate}
\inf_{f\in\Delta(\mathcal{S}_k)}\frac{r\alpha(f)+\log f(\ket{\psi})}{\log f(\ket{\varphi})},
\end{equation}
where this time $\Delta(\mathcal{S}_k)$ is the set of functions $f$ on pure unnormalised states that in addition to~\ref{it:tensormultiplicative}, \ref{it:tensoradditive} and~\ref{it:tensornormalised} satisfy for some $\alpha=\alpha(f)\in[0,1]$
\begin{enumerate}[(S0)]
\item\label{it:scaling} $f(\sqrt{p}\ket{\psi})=p^{\alpha}f(\ket{\psi})$
\item[(S1')]\label{it:LOCCmonotone} monotonicity under LOCC, or equivalently \cite[Theorem 3.1]{jensen2019asymptotic}
\begin{equation}\label{eq:monotonicitycondition}
f(\ket{\psi})\ge\left(f(\Pi_j\ket{\psi})^{1/\alpha}+f((I-\Pi_j)\ket{\psi})^{1/\alpha}\right)^{\alpha}
\end{equation}
(in a limit sense for $\alpha=0$) for every local orthogonal projecion $\Pi_j$ at party $j$.
\end{enumerate}
Note that the exponent can be recovered as $\alpha=\alpha(f)=\log f(\sqrt{2}\ket{0\ldots 0})$. We call $\Delta(\mathcal{S}_k)$ the asymptotic spectrum of LOCC transformations.

To provide some intuition on the relevant functions, it is helpful to study the bipartite case, where an explicit description of $\Delta(\mathcal{S}_2)$ is available \cite[Section 4.]{jensen2019asymptotic}. In this case for every $\alpha\in[0,1]$ there is precisely one monotone $f_\alpha\in\Delta(\mathcal{S}_2)$ with $\alpha(f_\alpha)=\alpha$, and its value on $\sum_i\sqrt{p_i}\ket{ii}$ is $\sum_{i}p_i^\alpha$. Another way to write this is
\begin{equation}\label{eq:bipartitemonotones}
f_\alpha(\ket{\psi})=2^{(1-\alpha)\entropy_\alpha(\Tr_2\ketbra{\psi}{\psi})}.
\end{equation}
For this reason we regard the monotones $f\in\Delta(\mathcal{S}_k)$ as generalisations of the R\'enyi entanglement entropies and $\alpha(f)$ as the generalisation of the order of the R\'enyi entropy.

An appealing feature of results like this is that they provide a bridge between the aforementioned operational and axiomatic approaches: they single out a list of axioms, and not only show that there exist monotones satisfying them simultaneously, but that in fact there are sufficiently many of them to characterise operational quantities (in this case transformation rates). In the context of ordered commutative monoids (as a mathematical model for general resource theories), a similar result is presented in \cite{fritz2017resource}. It should be emphasised that the proofs of both \eqref{eq:tensorrate} and \eqref{eq:LOCCrate} are non-constructive in the sense that they do not provide any explicit monotones in $\Delta(\mathcal{T}_k)$ (respectively $\Delta(\mathcal{S}_k)$). It is a nontrivial task to construct explicit monotones satisfying the respective axioms. A simple way to obtain new monotones from old ones is to compose a (possibly partial) flattening (i.e. grouping parties together) with an element of $\Delta(\mathcal{S}_{k'})$ for some $k'<k$. For $k'=2$ this construction gives the R\'enyi entropies (exponentiated as in \eqref{eq:bipartitemonotones}) of the reduced density matrices. We will refer to these as trivial points of $\Delta(\mathcal{S}_{k})$.

$\Delta(\mathcal{T}_k)$ can be identified with a subset of $\Delta(\mathcal{S}_k)$, namely $\alpha^{-1}(0)$. The first examples of nontrivial points in the asymptotic spectrum of tensors have been found in \cite{christandl2018universal}. In the present work we extend that construction and obtain nontrivial points in the asymptotic spectrum of LOCC transformations with $\alpha\neq 0$. The new family of monotones interpolates between the quantum functionals of \cite{christandl2018universal} and the (exponentiated) R\'enyi entropies of the single-party marginals. We prove the following properties of $\upperLOCC{\alpha,\theta}(\ket{\psi})=2^{\logupperLOCC{\alpha,\theta}(\ket{\psi})}$.
\begin{theorem}\label{thm:main}
Let $\alpha\in[0,1]$ and $\theta\in\distributions([k])$. For every $p\ge0$, $r\in\naturals$, $\ket{\psi}\in\mathcal{H}_1\otimes\cdots\otimes\mathcal{H}_k$ and $\ket{\varphi}\in\mathcal{K}_1\otimes\cdots\otimes\mathcal{K}_k$ the followings hold:
\begin{enumerate}[(i)]
\item\label{it:mainscaling} $\upperLOCC{\alpha,\theta}(\sqrt{p}\ket{\psi})=p^\alpha\upperLOCC{\alpha,\theta}(\ket{\psi})$,
\item\label{it:mainnormalised} $\upperLOCC{\alpha,\theta}(\unittensor{r})=r$,
\item\label{it:mainmultiplicative} $\upperLOCC{\alpha,\theta}(\ket{\psi}\otimes\ket{\varphi})=\upperLOCC{\alpha,\theta}(\ket{\psi})\upperLOCC{\alpha,\theta}(\ket{\varphi})$,
\item\label{it:mainadditive} $\upperLOCC{\alpha,\theta}(\ket{\psi}\oplus\ket{\varphi})=\upperLOCC{\alpha,\theta}(\ket{\psi})+\upperLOCC{\alpha,\theta}(\ket{\varphi})$,
\item\label{it:mainmonotone} if there exists a trace-nonincreasing LOCC channel $\Lambda$ such that $\Lambda(\ketbra{\psi}{\psi})=\ketbra{\varphi}{\varphi}$ then $\upperLOCC{\alpha,\theta}(\ket{\psi})\ge\upperLOCC{\alpha,\theta}(\ket{\varphi})$.
\end{enumerate}
\end{theorem}
That is, $\upperLOCC{\alpha,\theta}$ is a point in the asymptotic LOCC spectrum with $\alpha(\upperLOCC{\alpha,\theta})=\alpha$.

Before turning to the proof, let us attempt to argue why the existence of such an interpolating family is at least plausible. As mentioned above, for $\alpha=0$ these functionals reduce to the logarithmic quantum functionals defined in \cite{christandl2018universal} as the supremum $\entropy_{\theta}(\overline{\lambda})$ over the entanglement polytope. Remarkably, the R\'enyi entropies also admit a variational expression, namely \cite{arikan1996inequality,merhav1999shannon,shayevitz2011renyi}
\begin{equation}
(1-\alpha)\entropy_\alpha(P)=\sup_{Q\in\distributions(\mathcal{X})}\left[(1-\alpha)\entropy(Q)-\alpha\relativeentropy{Q}{P}\right],
\end{equation}
where $\relativeentropy{Q}{P}$ is the relative entropy between the probability distributions $Q$ and $P$. If we suspect that a common generalisation might exist, then the most straighforward idea is that inside the supremum the first term should become $(1-\alpha)\entropy_{\theta}$, whereas the relative entropy should be replaced with a function that is infinite outside the entanglement polytope and vanishes at the marginal spectra. These properties are satisfied by the rate function $\ratefunction{\overline{\lambda}}{\psi}$.

In addition, when $\alpha$ is nonzero and $\theta$ is concentrated on site $j$, we claim that \eqref{eq:multipartitemonotones} reduces to the exponentiated R\'enyi entropy of the $j$th marginal. To see this, note that in this case the first term in the supremum only depends on the $j$th component of $\overline{\lambda}$, therefore the optimisation over the remaining components can be carried out separately on the second term. According to the contraction principle, the infimum of the rate function over all but the $j$th argument is the rate function for the ordinary Keyl--Werner estimation at site $j$, i.e. the relative entropy distance from the ordered spectrum of the $j$th marginal.

For $\alpha\to 1$ all the functionals with different $\theta$ necessarily collapse to a single one, the norm squared \cite[Proposition 3.6]{jensen2019asymptotic}. However, as also suggested by \eqref{eq:bipartitemonotones}, the interesting limit at this point is (for normalised $\ket{\psi}$)
\begin{equation}\label{eq:alphaonelimit}
\lim_{\alpha\to 1}\frac{1}{1-\alpha}\logupperLOCC{\alpha,\theta}(\ket{\psi})=\sum_{j=1}^k\theta(j)\entropy(\ketbra{\psi}{\psi}_j),
\end{equation}
where $\ketbra{\psi}{\psi}_j$ is the $j$th marginal of the state. This is because in this limit the coefficient of the rate function goes to $-\infty$, penalising every point other than its unique zero, the collection of the marginal spectra. For this reason $\frac{1}{1-\alpha}\logupperLOCC{\alpha,\theta}(\ket{\psi})$ may be regarded as the R\'enyi generalisations of the limit \eqref{eq:alphaonelimit}.

The paper is structured as follows. In Section~\ref{sec:notations} we introduce our notations and collect some known results about the representation theory of symmetric and unitary groups. Our main focus here is on the asymptotic dimension of irreducible representations and vanishing conditions for multiplicities. These can be expressed in terms of limits of rescaled integer partitions (normalised decreasing nonnegative real sequences) and moment polytopes. Section~\ref{sec:estimation} studies the rate function of a multiparty version of the spectrum estimation scheme of Keyl and Werner \cite{keyl2001estimating}. In Section~\ref{sec:spectralpoints} we prove our main result, Theorem~\ref{thm:main}.

\section{Notations and preliminaries}\label{sec:notations}

Most of what follows is standard material in the representation theory of classical groups and can be found in many textbooks, see e.g. \cite{fulton1991representation}. For an introduction aimed at quantum information theorists and emphasising the asymptotic aspects we refer to Hayashi's book \cite{hayashi2017group}.

We will denote by $\partitions[n]$ the set of partitions of the integer $n$, i.e. nonincreasing nonnegative integer sequences summing to $n$. Partitions will serve as labels of irreducible representations of the unitary and the symmetric groups. In particular, if $\mathcal{H}$ is a finite dimensional Hilbert space then $\mathcal{H}^{\otimes n}$ has the Schur--Weyl decomposition
\begin{equation}
\mathcal{H}^{\otimes n}\simeq\bigoplus_{\alpha\in\partitions[n]}\mathbb{S}_{\alpha}(\mathcal{H})\otimes[\alpha],
\end{equation}
which is an isomorphism of $U(\mathcal{H})\times S_n$-representations. Here $\mathbb{S}_{\alpha}(\mathcal{H})$ is the irreducible representation of $U(\mathcal{H})$ with highest weight $\alpha$ if $\alpha$ has at most $\dim\mathcal{H}$ parts and the zero representation otherwise, while $[\alpha]$ is an irreducible representation of $S_n$. The orthogonal projection onto the isotypic component corresponding to $\alpha$ will be denoted by $P^{\mathcal{H}}_\alpha$ or $P_\alpha$ if the Hilbert space is clear from the context. The number of partitions of $n$ into at most $\dim\mathcal{H}$ parts is upper bounded by $(n+1)^{\dim\mathcal{H}}$. We need the following dimension estimates
\begin{align}
\dim\mathbb{S}_{\alpha}(\mathcal{H}) & \le(|\alpha|+1)^{d(d-1)/2}  \label{eq:Schurdimension}  \\
\dim[\alpha] & \ge \frac{1}{(|\alpha|+d)^{(d+2)(d-1)/2}}2^{|\alpha|\entropy(\frac{1}{|\alpha|}\alpha)}  \label{eq:symmetricgrpdimension}
\end{align}
where $d=\dim\mathcal{H}$.

We denote by $\overline{\partitions}$ the set of those functions $\naturals\to\nonnegativereals$ that are nonincreasing, finitely supported and that sum to $1$. In particular, when $\alpha\in\partitions[n]$, then we can consider its normalisation $\frac{1}{n}\alpha\in\overline{\partitions}$, where the multiplication is understood entrywise. $\overline{\partitions}$ is equipped with the metric induced by the $\ell^1$ norm, the open ball of radius $\epsilon$ around $\overline{\alpha}\in\overline{\partitions}$ is $\ball{\epsilon}{\overline{\alpha}}$. An element $\overline{\alpha}\in\overline{\partitions}$ can be viewed as a probability distribution and we can consider its Shannon entropy $\entropy(\overline{\alpha})$. When $\rho$ is a quantum state on a finite dimensional Hilbert space, then its spectrum (with multiplicities and ordered nonincreasingly) is an element of $\overline{\partitions}$.

Let $\alpha,\beta,\gamma\in\partitions[n]$. The Kronecker coefficient is defined as $g_{\alpha\beta\gamma}=\dim([\alpha]\otimes[\beta]\otimes[\gamma])^{S_n}$. We define $\Kron$ to be the closure in $\overline{\partitions}^3$ of
\begin{equation}
\setbuild{\left(\frac{1}{n}\alpha,\frac{1}{n}\beta,\frac{1}{n}\gamma\right)}{|\alpha|=|\beta|=|\gamma|=n,g_{\alpha\beta\gamma}\neq 0}.
\end{equation}
$\Kron$ is also the set of triples that arise as marginal spectra of tripartite pure states. $(\overline{\alpha},\overline{\beta},\overline{\gamma})\in\Kron$ implies $\entropy(\overline{\alpha})\le\entropy(\overline{\beta})+\entropy(\overline{\gamma})$ \cite{klyachko2004quantum,christandl2006spectra,christandl2006structure}.

The Kronecker coefficients also appear in decompositions for the unitary groups. Let $\mathcal{H}$ and $\mathcal{K}$ be Hilbert spaces and $\alpha\in\partitions[n]$. Then there is a map $U(\mathcal{H})\times U(\mathcal{K})\to U(\mathcal{H}\otimes\mathcal{K})$ (given by the tensor product of the unitary operators) and we have the isomorphism
\begin{equation}
\mathbb{S}_{\alpha}(\mathcal{H}\otimes\mathcal{K})\simeq\bigoplus_{\beta,\gamma\in\partitions[n]}\complexes^{g_{\alpha\beta\gamma}}\otimes\mathbb{S}_{\beta}(\mathcal{H})\otimes\mathbb{S}_{\gamma}(\mathcal{K})
\end{equation}
as $U(\mathcal{H})\times U(\mathcal{K})$-representations. In terms of the projections $P^{\mathcal{H}\otimes\mathcal{K}}_\alpha$, $P^{\mathcal{H}}_\beta\otimes\id_{\mathcal{K}^{\otimes n}}$ and $\id_{\mathcal{H}^{\otimes n}}\otimes P^{\mathcal{K}}_\gamma$ on $(\mathcal{H}\otimes\mathcal{K})^{\otimes n}$ we have the vanishing condition (here and below assuming that the individual projections do not vanish)
\begin{equation}\label{eq:Kronvanishing}
P^{\mathcal{H}\otimes\mathcal{K}}_\alpha(P^{\mathcal{H}}_\beta\otimes P^{\mathcal{K}}_\gamma)=0\iff g_{\alpha\beta\gamma}=0.
\end{equation}
Note that the projections $P^{\mathcal{H}\otimes\mathcal{K}}_\alpha$ and $P^{\mathcal{H}}_\beta\otimes P^{\mathcal{K}}_\gamma$ commute.

Let $0\le m\le n$ and $\alpha\in\partitions[n],\beta\in\partitions[m],\gamma\in\partitions[n-m]$. The Littlewood--Richardson coefficient is defined as $c^\alpha_{\beta\gamma}=\dim(([\beta]\otimes[\gamma])\otimes\Res^{S_{n}}_{S_m\times S_{n-m}}[\alpha])^{S_m\times S_{n-m}}$. We define $\LR{\bullet}$ to be the closure in $[0,1]\times\overline{\partitions}^3$ of
\begin{equation}
\setbuild{\left(\frac{m}{n},\frac{1}{n}\alpha,\frac{1}{m}\beta,\frac{1}{n-m}\gamma\right)}{|\alpha|=n,|\beta|=m>0,|\gamma|=n-m>0,c^\alpha_{\beta\gamma}\neq 0}
\end{equation}
and $\LR{q}$ as
\begin{equation}
\LR{q}=\setbuild{(\overline{\alpha},\overline{\beta},\overline{\gamma})}{(q,\overline{\alpha},\overline{\beta},\overline{\gamma})\in\LR{\bullet}}
\end{equation}
for $q\in[0,1]$. $\LR{q}$ is also the set of triples $(\overline{\alpha},\overline{\beta},\overline{\gamma})$ such that there exist quantum states $\rho,\sigma$ on a finite dimensional Hilbert space such that the spectra of $q\rho+(1-q)\sigma,\rho,\sigma$ are $\overline{\alpha},\overline{\beta},\overline{\gamma}$, respectively, which is a form of Horn's problem \cite{horn1962eigenvalues,lidskii1982spectral,klyachko1998stable,christandl2006structure}. $(\overline{\alpha},\overline{\beta},\overline{\gamma})\in\LR{q}$ implies $q\entropy(\overline{\beta})+(1-q)\entropy(\overline{\gamma})\le\entropy(\overline{\alpha})\le q\entropy(\overline{\beta})+(1-q)\entropy(\overline{\gamma})+h(q)$ \cite{christandl2018universal}.

The Littlewood--Richardson coefficients also appear in decompositions for the unitary groups in two ways. Let $\mathcal{H}$ and $\mathcal{K}$ be Hilbert spaces and $\alpha\in\partitions[n]$. Then $U(\mathcal{H})\times U(\mathcal{K})\le U(\mathcal{H}\oplus\mathcal{K})$ (as block diagonal operators) and we have the isomorphism
\begin{equation}
\mathbb{S}_\alpha(\mathcal{H}\oplus\mathcal{K})\simeq\bigoplus_{\beta,\gamma}\complexes^{c^\alpha_{\beta\gamma}}\otimes\mathbb{S}_{\beta}(\mathcal{H})\otimes\mathbb{S}_{\gamma}(\mathcal{K})
\end{equation}
as $U(\mathcal{H})\times U(\mathcal{K})$-representations, where the sum is over partitions with $|\beta|+|\gamma|=n=|\alpha|$. Consider the representation on $(\mathcal{H}\oplus\mathcal{K})^{\otimes n}\simeq\bigoplus_{m=0}^{n}\complexes^{\binom{n}{m}}\otimes\mathcal{H}^{\otimes m}\otimes\mathcal{K}^{\otimes n-m}$. In the direct summands the first factor corresponds to the possible $m$-element subsets of the $n$ factors where the space $\mathcal{H}$ is chosen. In general, to specify a vector or an operator on this space in terms of this isomorphism, one needs to choose a bijection between the basis elements of $\complexes^{\binom{n}{m}}$ and the subsets, and in addition an ordering of the $m$ and $n-m$ factors among themselves. However, we will only encounter instances where the vectors or operators are $S_m\times S_{n-m}$-invariant and we sum over all the $m$-element subsets, which eliminates the need for these choices. In particular, we consider the commuting projections $P^{\mathcal{H}\oplus\mathcal{K}}_\alpha$, $\id_{\complexes^{\binom{n}{m}}}\otimes P^{\mathcal{H}}_\beta\otimes\id_{\mathcal{K}^{\otimes n-m}}$ and $\id_{\complexes^{\binom{n}{m}}}\otimes\id_{\mathcal{H}^{\otimes m}}\otimes P^{\mathcal{K}}_\gamma$, in terms of which we have the vanishing condition
\begin{equation}\label{eq:LRvanishingsum}
P^{\mathcal{H}\oplus\mathcal{K}}_\alpha(\id_{\complexes^{\binom{n}{m}}}\otimes P^{\mathcal{H}}_\beta\otimes P^{\mathcal{K}}_\gamma)=0\iff c^\alpha_{\beta\gamma}=0.
\end{equation}

The second decomposition involving the Littlewood--Richardson coefficients is that of the tensor product of representations of a unitary group $U(\mathcal{H})$. For $\beta\in\partitions[m]$ and $\gamma\in\partitions[n-m]$ we have the isomorphism
\begin{equation}
\mathbb{S}_\beta(\mathcal{H})\otimes\mathbb{S}_\gamma(\mathcal{H})\simeq\bigoplus_{\alpha\in\partitions[n]}\mathbb{S}_\alpha(\mathcal{H})
\end{equation}
as $U(\mathcal{H})$-representations. Choosing a factorisation $\mathcal{H}^{\otimes n}\simeq\mathcal{H}^{\otimes m}\otimes\mathcal{H}^{\otimes n-m}$ we may consider the projections $P^{\mathcal{H}}_\alpha$, $P^{\mathcal{H}}_\beta\otimes\id_{\mathcal{H}^{\otimes n-m}}$ and $\id_{\mathcal{H}^{\otimes m}}\otimes P^{\mathcal{H}}_\gamma$. These projections commute and satisfy the vanishing condition
\begin{equation}\label{eq:LRvanishingproduct}
P^{\mathcal{H}}_\alpha(P^{\mathcal{H}}_\beta\otimes P^{\mathcal{H}}_\gamma)=0\iff c^\alpha_{\beta\gamma}=0.
\end{equation}

Next we consider Hilbert spaces of multipartite systems and introduce a compact notation in order to simplify the formulas later. $k$ denotes the number of subsystems (tensor factors), which can be considered fixed throughout.  We use $\lambda,\mu,\nu$ to denote $k$-tuples of partitions of some natural number $n$ and write $\lambda=(\lambda_1,\ldots,\lambda_k)$ when referring to the individual partitions (we will not need to label the parts of the partitions). If $\mathcal{H}=\mathcal{H}_1\otimes\cdots\otimes\mathcal{H}_k$, then we can decompose each factor in $\mathcal{H}^{\otimes n}$ using the Schur--Weyl decomposition
\begin{equation}
(\mathcal{H}_1\otimes\cdots\otimes\mathcal{H}_k)^{\otimes n}\simeq\bigoplus_{\lambda\in\partitions[n]^k}\bigotimes_{j=1}^k\mathbb{S}_{\lambda_j}(\mathcal{H}_j)\otimes[\lambda_j].
\end{equation}
This is an isomorphism of $U(\mathcal{H}_1)\times\cdots\times U(\mathcal{H}_k)\times S_n^k$-representations. A tuple of partitions $\lambda$ can be identified with a dominant weight for the compact Lie group $U(\mathcal{H}_1)\times\cdots\times U(\mathcal{H}_k)$ (there are other dominant weights, but these are the ones that we shall encounter). $P^{\mathcal{H}}_{\lambda}$ (also $P_{\lambda}$ if the Hilbert space is understood) will denote the orthogonal projection onto the direct summand corresponding to $\lambda$. It satisfies
\begin{equation}\label{eq:localprojectorproduct}
P^{\mathcal{H}}_{\lambda}=P^{\mathcal{H}_1}_{\lambda_1}\otimes\cdots\otimes P^{\mathcal{H}_k}_{\lambda_k}.
\end{equation}

When $\lambda\in\partitions[n]^k$, we will also write $\frac{1}{n}\lambda\in\overline{\partitions}^k$ for $(\frac{1}{n}\lambda_1,\ldots,\frac{1}{n}\lambda_k)$, meaning that every element of every partition in the $k$-tuple is rescaled. On $\overline{\partitions}^k$ we consider the distance induced by the maximum of the $1$-norms. For $\overline{\lambda}=(\overline{\lambda_1},\ldots,\overline{\lambda_k})\in\overline{\partitions}^k$ and a probability distribution $\theta\in\distributions([k])$ we will frequently use the weighted average of their entropies
\begin{equation}
\entropy_{\theta}(\overline{\lambda})=\sum_{j=1}^k\theta(j)\entropy(\overline{\lambda_j}).
\end{equation}

When $\overline{\lambda},\overline{\mu},\overline{\nu}\in\overline{\partitions}^k$, we will write $(\overline{\lambda},\overline{\mu},\overline{\nu})\in\Kron^k$ to mean $\forall j\in[k]:(\overline{\lambda_j},\overline{\mu_j},\overline{\nu_j})\in\Kron$. By taking convex combinations we can see that $(\overline{\lambda},\overline{\mu},\overline{\nu})\in\Kron^k$ implies
\begin{equation}
\entropy_{\theta}(\overline{\lambda})\le\entropy_{\theta}(\overline{\mu})+\entropy_{\theta}(\overline{\nu}).
\end{equation}
Similarly, when $\lambda,\mu,\nu\in\partitions[n]^k$ we will use $g_{\lambda\mu\nu}\neq 0$ as an abbreviation for $\forall j\in[k]:g_{\lambda_j\mu_j\nu_j}\neq 0$ (one could define $g_{\lambda\mu\nu}=\prod_{j=1}^k g_{\lambda_j\mu_j\nu_j}$ but this value will not play a role, the only thing that matters is if it is zero or not).

When $\overline{\lambda},\overline{\mu},\overline{\nu}\in\overline{\partitions}^k$, we will write $(\overline{\lambda},\overline{\mu},\overline{\nu})\in\LR{q}^k$ to mean $\forall j\in[k]:(\overline{\lambda_j},\overline{\mu_j},\overline{\nu_j})\in\LR{q}$. $(\overline{\lambda},\overline{\mu},\overline{\nu})\in\LR{q}^k$ implies
\begin{equation}
q\entropy_{\theta}(\overline{\mu})+(1-q)\entropy_{\theta}(\overline{\nu})\le\entropy_{\theta}(\overline{\lambda})\le q\entropy_{\theta}(\overline{\mu})+(1-q)\entropy_{\theta}(\overline{\nu})+h(q).
\end{equation}
Similarly, when $\lambda\in\partitions[n]^k,\mu\in\partitions[m]^k,\nu\in\partitions[n-m]^k$ we will use $c^\lambda_{\mu\nu}\neq 0$ as an abbreviation for $\forall j\in[k]:c^{\lambda_j}_{\mu_j\nu_j}\neq 0$ (again, one may think $c^{\lambda}_{\mu\nu}=\prod_{j=1}^k c^{\lambda_j}_{\mu_j\nu_j}$ but the value itself will not be used).

\section{Simultaneous spectrum estimation}\label{sec:estimation}

In \cite{alicki1988symmetry} Alicki, Rudicki and Sadowski and in \cite{keyl2001estimating} Keyl and Werner proposed an estimator for the spectrum of a density matrix, based on measuring the Schur--Weyl projectors $P_\alpha$ on $n$ identical copies of a quantum state. The measurement outcomes are labelled with the partitions $\alpha\in\partitions[n]$ and $\frac{1}{n}\alpha\in\overline{\partitions}$ is the estimate for the ordered spectrum. They showed that as $n\to\infty$, the distributions converge weakly to the Dirac measure on the spectrum and satisfy a large deviation principle with rate function given by the relative entropy. For our purposes the following formulation will be convenient: if $r$ denotes the ordered spectrum of $\rho$ and $\overline{\alpha}\in\overline{\partitions}$ then
\begin{equation}\label{eq:KeylWerner}
\lim_{\epsilon\to0}\lim_{n\to\infty}-\frac{1}{n}\log\sum_{\substack{\alpha\in\partitions[n]  \\  \frac{1}{n}\alpha\in\ball{\epsilon}{\overline{\alpha}}}}\Tr(P_\alpha\rho^{\otimes n})=\relativeentropy{\overline{\alpha}}{r}.
\end{equation}
Suppose that $\ketbra{\varphi}{\varphi}\in\states(\mathcal{H}_1\otimes\mathcal{H}_2)$ is a purification of $\rho\in\states(\mathcal{H}_1)$. If we perform the measurement independently on $\mathcal{H}_1^{\otimes n}$ and on $\mathcal{H}_2^{\otimes n}$ then it is easy to see that the outcomes will be perfectly correlated and of course both sides see the same exponential behaviour as in \eqref{eq:KeylWerner}.

In this section we will study the rate function for the same estimator in a multipartite setting. Let $\varphi\in\mathcal{H}=\mathcal{H}_1\otimes\cdots\otimes\mathcal{H}_k$ be a unit vector and suppose that the $k$ parties perform a measurement with the local Schur--Weyl projectors $P^{\mathcal{H}_j}_{\lambda_j}$ on $\varphi^{\otimes n}$. Note that this is equivalent to measuring $P^{\mathcal{H}}_{\lambda}$ (see \eqref{eq:localprojectorproduct}). The rescaled outcome $\frac{1}{n}\lambda\in\overline{\partitions}^k$ serves as the estimate of the $k$-tuple of marginal spectra. As in the bipartite case, each marginal estimate looks like \eqref{eq:KeylWerner} but this time the correlation between the estimates is more complicated. We regard the resulting rate function as a multipartite generalisation of the relative entropy and it will play a central role in our construction of the entanglement monotones. Recall that the classical relative entropy satisfies 
\begin{equation}\label{eq:relativeentropyproduct}
\relativeentropy{Q}{P_1\otimes P_2}\ge\relativeentropy{Q_1}{P_1}+\relativeentropy{Q_2}{P_2}
\end{equation}
where $Q_1$ and $Q_2$ are the marginals of $Q$, and
\begin{equation}\label{eq:relativeentropysum}
\relativeentropy{qQ_1\oplus(1-q)Q_2}{P_1\oplus P_2}=q\relativeentropy{Q_1}{P_1}+(1-q)\relativeentropy{Q_2}{P_2}-h(q),
\end{equation}
where $Q_1\in\distributions(\mathcal{X}_1)$, $Q_2\in\distributions(\mathcal{X}_2)$, $\oplus$ is the direct sum resulting in a distribution on $\mathcal{X}_1\cup\mathcal{X}_2$ and $h(p)=-p\log p-(1-p)\log(1-p)$. The results in this section can be viewed as analogous properties satisfied by the rate function in the simultaneous spectrum estimation problem and may be of independent interest.

We make the following definition:
\begin{definition}\label{def:ratefunction}
Let $\overline{\lambda}\in\overline{\partitions}^k$ and $\varphi\in\mathcal{H}=\mathcal{H}_1\otimes\cdots\otimes\mathcal{H}_k$. The rate function is defined as
\begin{equation}
\ratefunction{\overline{\lambda}}{\varphi}=\lim_{\epsilon\to 0}\lim_{n\to\infty}-\frac{1}{n}\log\sum_{\substack{\lambda\in\partitions[n]^k  \\  \frac{1}{n}\lambda\in \ball{\epsilon}{\overline{\lambda}}}}\norm{P_\lambda\varphi^{\otimes n}}^2
\end{equation}
\end{definition}
To see that this quantity is well defined (possibly $\infty$), we only need to show that the limit as $n\to\infty$ exists, which is then clearly monotone in $\epsilon$. We postpone the proof of this fact (Proposition~\ref{prop:limballexists}, see also \cite{botero2020large}) and of the following technical lemmas to Section~\ref{sec:technical}.
\begin{lemma}\label{lem:subsequenceliminfbound}
Let $\varphi\in\mathcal{H}$, $(n_l)_{l\in\naturals}$ a sequence of natural numbers such that $\lim_{l\to\infty}n_l=\infty$ and $(\lambda^{(n_l)})_{l\in\naturals}$ a sequence such that $\lambda^{(n_l)}\in\partitions[n_l]^k$ and $\lim_{l\to\infty}\frac{1}{n_l}\lambda^{(n_l)}=\overline{\lambda}$. Then
\begin{equation}
\liminf_{l\to\infty}-\frac{1}{n_l}\log\norm{P_{\lambda^{(n_l)}}\varphi^{\otimes n_l}}^2\ge\ratefunction{\overline{\lambda}}{\varphi}.
\end{equation}
\end{lemma}
\begin{lemma}\label{lem:ratefunctionsinglelambdalimit}
For every $\varphi\in\mathcal{H}$ and $\overline{\lambda}\in\overline{\partitions}^k$ there exists a sequence $(\lambda^{(n)})_{n\in\naturals}$ such that $\lambda^{(n)}\in\partitions[n]^k$,
\begin{align}
\lim_{n\to\infty}\frac{1}{n}\lambda^{(n)} & = \overline{\lambda}
\intertext{and}
\lim_{n\to\infty}-\frac{1}{n}\log\norm{P_{\lambda^{(n)}}\varphi^{\otimes n}}^2 & = \ratefunction{\overline{\lambda}}{\varphi}.
\end{align}
\end{lemma}

We mention that the rate function can be expressed via a single-letter formula as a special case of \cite{franks2020minimal,botero2020large}:
\begin{equation}
\ratefunction{\overline{\lambda}}{\varphi}=\inf_U\sup_{A,N}\langle\overline{\lambda},\alpha\rangle-\log\Tr N^*2^{A/2}U^*\ketbra{\varphi}{\varphi}U2^{A/2}N,
\end{equation}
where the infimum is over tensor product unitaries $U$, the supremum is over $A\in\reals^{\dim\mathcal{H}_1}\oplus\cdots\oplus\reals^{\dim\mathcal{H}_k}$, identified with diagonal matrices as $A_1\otimes I\otimes\cdots\otimes I+I\otimes A_2\otimes I\otimes\cdots\otimes I+\cdots$ and over tensor products $N$ of upper triangular unipotent matrices. We will not make use of this expression in the proofs below, but we expect it to be useful for computations.

First we derive some immediate properties of the rate function that will be used in Section~\ref{sec:spectralpoints}.
\begin{proposition}[Basic properties of the rate function]\label{prop:ratebasic}
Let $\varphi\in\mathcal{H}=\mathcal{H}_1\otimes\cdots\otimes\mathcal{H}_k$.
\begin{enumerate}[(i)]
\item\label{it:ratescaling} $\ratefunction{\overline{\lambda}}{\sqrt{p}\varphi}=\ratefunction{\overline{\lambda}}{\varphi}-\log p$ for every $\overline{\lambda}\in\overline{\partitions}^k$ and $p>0$.
\item\label{it:ratesimplemonotone} If $\psi=(A_1\otimes\cdots\otimes A_k)\varphi$ where $\forall j:A_j\in\boundeds(\mathcal{H}_j)$ and $A_j^*A_j\le I$ then $\ratefunction{\overline{\lambda}}{\varphi}\le\ratefunction{\overline{\lambda}}{\psi}$ for every $\overline{\lambda}\in\overline{\partitions}^k$
\item\label{it:ratemarginal} If $\norm{\varphi}=1$ and $\overline{\lambda}$ is the collection of its marginal spectra then $\ratefunction{\overline{\lambda}}{\varphi}=0$ (see also \cite[Corollary 3.25.]{botero2020large}).
\end{enumerate}
\end{proposition}
\begin{proof}
\ref{it:ratescaling}:
This follows from Definition~\ref{def:ratefunction} and the equality
\begin{equation}
-\frac{1}{n}\log\sum_{\substack{\lambda\in\partitions[n]^k  \\  \frac{1}{n}\lambda\in \ball{\epsilon}{\overline{\lambda}}}}\norm{P_\lambda(\sqrt{p}\varphi)^{\otimes n}}^2=-\frac{1}{n}\log\sum_{\substack{\lambda\in\partitions[n]^k  \\  \frac{1}{n}\lambda\in \ball{\epsilon}{\overline{\lambda}}}}\norm{P_\lambda\varphi^{\otimes n}}^2-\log p.
\end{equation}

\ref{it:ratesimplemonotone}:
$(A_1\otimes\cdots\otimes A_k)^{\otimes n}$ commutes with $P_\lambda$, therefore
\begin{equation}
\begin{split}
\norm{P_\lambda\left((A_1\otimes\cdots\otimes A_k)\varphi\right)^{\otimes n}}^2
 & = \norm{(A_1\otimes\cdots\otimes A_k)^{\otimes n}P_\lambda\varphi^{\otimes n}}^2  \\
 & = \langle P_\lambda\varphi^{\otimes n},(A_1^*A_1\otimes\cdots\otimes A_k^*A_k)^{\otimes n}P_\lambda\varphi^{\otimes n}\rangle  \\
 & \le \norm{P_\lambda\varphi^{\otimes n}}^2.
\end{split}
\end{equation}
From this the statement follows using Definition~\ref{def:ratefunction}.

\ref{it:ratemarginal}: Let $\rho_j$ be the $j$th marginal of $\ketbra{\varphi}{\varphi}$. \eqref{eq:KeylWerner} implies that for every $\epsilon>0$ and $j$ we have
\begin{equation}
\lim_{n\to\infty}\sum_{\alpha\in\partitions[n]\cap n\ball{\epsilon}{\overline{\lambda_j}}}\Tr P_\alpha\rho_j^{\otimes n}=1,
\end{equation}
which implies
\begin{equation}
\lim_{n\to\infty}\sum_{\substack{\lambda\in\partitions[n]^k  \\  \frac{1}{n}\lambda\in \ball{\epsilon}{\overline{\lambda}}}}\norm{P_\lambda\varphi^{\otimes n}}^2=1.
\end{equation}
Therefore the limit in Definition~\ref{def:ratefunction} is $0$ (even without dividing by $n$).
\end{proof}

The following inequality is analogous to \eqref{eq:relativeentropyproduct} and will be used in the proof of submultiplicativity in Proposition~\ref{prop:uppersub}.
\begin{proposition}[Rate function and tensor product]\label{prop:rateproduct}
Let $\psi\in\mathcal{H}=\mathcal{H}_1\otimes\cdots\otimes\mathcal{H}_k$ and $\varphi\in\mathcal{K}=\mathcal{K}_1\otimes\cdots\otimes\mathcal{K}_k$. For every $\overline{\lambda}\in\overline{\partitions}^k$ the inequality
\begin{equation}
\ratefunction{\overline{\lambda}}{\psi\otimes\varphi}\ge\inf_{\substack{\overline{\mu},\overline{\nu}\in\overline{\partitions}^k  \\  (\overline{\lambda},\overline{\mu},\overline{\nu})\in\Kron^k}}\ratefunction{\overline{\mu}}{\psi}+\ratefunction{\overline{\nu}}{\varphi}
\end{equation}
holds.
\end{proposition}
\begin{proof}
Let $d=\sum_{j=1}^k(\dim\mathcal{H}_j+\dim\mathcal{K}_j)$ and $\lambda\in\partitions[n]^k$. Using that the sum of Schur--Weyl projections is the identity, the estimate on the number of partitions with bounded length and the vanishing condition \eqref{eq:Kronvanishing} we have the inequality
\begin{equation}\label{eq:tensorproductmainestimate}
\begin{split}
\norm{P^{\mathcal{H}\otimes\mathcal{K}}_{\lambda}(\psi\otimes\varphi)^{\otimes n}}^2
 & = \sum_{\substack{\mu,\nu\in\partitions[n]^k  \\  g_{\lambda\mu\nu}\neq 0}}\norm{P^{\mathcal{H}\otimes\mathcal{K}}_{\lambda}(P^{\mathcal{H}}_{\mu}\psi^{\otimes n}\otimes P^{\mathcal{K}}_{\nu}\varphi^{\otimes n})}^2  \\
 & \le \sum_{\substack{\mu,\nu\in\partitions[n]^k  \\  g_{\lambda\mu\nu}\neq 0}}\norm{P^{\mathcal{H}}_{\mu}\psi^{\otimes n}\otimes P^{\mathcal{K}}_{\nu}\varphi^{\otimes n}}^2  \\
 & \le (n+1)^d\max_{\substack{\mu,\nu\in\partitions[n]^k  \\  g_{\lambda\mu\nu}\neq 0}}\norm{P^{\mathcal{H}}_{\mu}\psi^{\otimes n}}^2\norm{P^{\mathcal{K}}_{\nu}\varphi^{\otimes n}}^2
\end{split}
\end{equation}
By Lemma~\ref{lem:ratefunctionsinglelambdalimit} we can choose a sequence $\lambda^{(1)},\lambda^{(2)},\ldots$ such that $\lambda^{(n)}\in\partitions[n]^k$, $\lim_{n\to\infty}\frac{1}{n}\lambda^{(n)}=\overline{\lambda}$ and
\begin{equation}
\lim_{n\to\infty}-\frac{1}{n}\log\norm{P^{\mathcal{H}\otimes\mathcal{K}}_{\lambda}(\psi\otimes\varphi)^{\otimes n}}^2=\ratefunction{\overline{\lambda}}{\psi\otimes\varphi}.
\end{equation}
For every $n$ choose $\mu^{(n)},\nu^{(n)}\in\partitions[n]^k$ such that $g_{\lambda^{(n)}\mu^{(n)}\nu^{(n)}}\neq 0$ and attaining the maximum in \eqref{eq:tensorproductmainestimate}. Choose a subsequence $n_l$ such that both $\frac{1}{n_l}\mu^{(n_l)}$ and $\frac{1}{n_l}\nu^{(n_l)}$ converge (possible since finite dimensional slices of $\overline{\partitions}$ are compact), and let their limits be $\overline{\mu},\overline{\nu}$. $\Kron^k$ is closed, therefore $(\overline{\lambda},\overline{\mu},\overline{\nu})\in\Kron^k$.
\begin{equation}
\begin{split}
\ratefunction{\overline{\lambda}}{\psi\otimes\varphi}
 & = \lim_{n\to\infty}-\frac{1}{n}\log\norm{P^{\mathcal{H}\otimes\mathcal{K}}_{\lambda^{(n)}}(\psi\otimes\varphi)^{\otimes n}}^2  \\
 & = \lim_{l\to\infty}-\frac{1}{n_l}\log\norm{P^{\mathcal{H}\otimes\mathcal{K}}_{\lambda^{(n_l)}}(\psi\otimes\varphi)^{\otimes n_l}}^2  \\
 & \ge \liminf_{l\to\infty}-\frac{1}{n_l}\log\norm{P^{\mathcal{H}}_{\mu^{(n_l)}}\psi^{\otimes n_l}}^2-\frac{1}{n_l}\log\norm{P^{\mathcal{K}}_{\nu^{(n_l)}}\varphi^{\otimes n_l}}^2  \\
 & \ge \ratefunction{\overline{\mu}}{\psi}+\ratefunction{\overline{\nu}}{\varphi}.
\end{split}
\end{equation}
The first inequality follows from \eqref{eq:tensorproductmainestimate} and the second one from Lemma~\ref{lem:subsequenceliminfbound}.
\end{proof}

The following two inequalities should be compared with \eqref{eq:relativeentropysum}. In Section~\ref{sec:spectralpoints} the first one  (Proposition~\ref{prop:ratesum}) will be used in the proof of additivity of our monotones, while the second one (Proposition~\ref{prop:rateprojection}) is needed in the proof of monotonicity.
\begin{proposition}[Rate function and direct sum]\label{prop:ratesum}
Let $\psi\in\mathcal{H}=\mathcal{H}_1\otimes\cdots\otimes\mathcal{H}_k$ and $\varphi\in\mathcal{K}=\mathcal{K}_1\otimes\cdots\otimes\mathcal{K}_k$. For every $\overline{\lambda}\in\overline{\partitions}^k$ the inequality
\begin{equation}
\ratefunction{\overline{\lambda}}{\psi\oplus\varphi}\ge\inf_{q\in[0,1]}\inf_{\substack{\overline{\mu},\overline{\nu}\in\overline{\partitions}^k  \\  (\overline{\lambda},\overline{\mu},\overline{\nu})\in\LR{q}^k}}q\ratefunction{\overline{\mu}}{\psi}+(1-q)\ratefunction{\overline{\nu}}{\varphi}-h(q)
\end{equation}
holds.
\end{proposition}
\begin{proof}
Let $d=1+\sum_{j=1}^k(\dim\mathcal{H}_j+\dim\mathcal{K}_j)$ and $\lambda\in\partitions[n]^k$. Using that the sum of Schur--Weyl projections is the identity, the estimate on the number of partitions with bounded length and the vanishing condition \eqref{eq:LRvanishingsum} we have the inequality
\begin{equation}\label{eq:directsummainestimate}
\begin{split}
\norm{P^{\mathcal{H}\oplus\mathcal{K}}_{\lambda}(\psi\oplus\varphi)^{\otimes n}}^2
 & = \sum_{m=0}^n\sum_{\substack{\mu\in\partitions[m]^k  \\  \nu\in\partitions[n-m]^k  \\  c^{\lambda}_{\mu\nu}\neq 0}}\norm{P^{\mathcal{H}\oplus\mathcal{K}}_\lambda(\id_{\complexes^{\binom{n}{m}}}\otimes P^{\mathcal{H}}_\mu\otimes P^{\mathcal{K}}_\nu)\left(\unittensor{\binom{n}{m}}\otimes\psi^{\otimes m}\otimes\varphi^{\otimes n-m}\right)}^2  \\
 & \le \sum_{m=0}^n\sum_{\substack{\mu\in\partitions[m]^k  \\  \nu\in\partitions[n-m]^k  \\  c^{\lambda}_{\mu\nu}\neq 0}}\norm{(\id_{\complexes^{\binom{n}{m}}}\otimes P^{\mathcal{H}}_\mu\otimes P^{\mathcal{K}}_\nu)\left(\unittensor{\binom{n}{m}}\otimes\psi^{\otimes m}\otimes\varphi^{\otimes n-m}\right)}^2  \\
 & = \sum_{m=0}^n\binom{n}{m}\sum_{\substack{\mu\in\partitions[m]^k  \\  \nu\in\partitions[n-m]^k  \\  c^{\lambda}_{\mu\nu}\neq 0}}\norm{P^{\mathcal{H}}_{\mu}\psi^{\otimes m}}^2\norm{P^{\mathcal{K}}_{\nu}\varphi^{\otimes n-m}}^2  \\
 & \le (n+1)^d\max_{\substack{0\le m\le n  \\  \mu\in\partitions[m]^k  \\  \nu\in\partitions[n-m]^k  \\  c^{\lambda}_{\mu\nu}\neq 0}}\binom{n}{m}\norm{P^{\mathcal{H}}_{\mu}\psi^{\otimes m}}^2\norm{P^{\mathcal{K}}_{\nu}\varphi^{\otimes n-m}}^2
\end{split}
\end{equation}
By Lemma~\ref{lem:ratefunctionsinglelambdalimit} we can choose a sequence $\lambda^{(1)},\lambda^{(2)},\ldots$ such that $\lambda^{(n)}\in\partitions[n]^k$, $\lim_{n\to\infty}\frac{1}{n}\lambda^{(n)}=\overline{\lambda}$ and
\begin{equation}
\lim_{n\to\infty}-\frac{1}{n}\log\norm{P^{\mathcal{H}\oplus\mathcal{K}}_{\lambda}(\psi\oplus\varphi)^{\otimes n}}^2=\ratefunction{\overline{\lambda}}{\psi\oplus\varphi}.
\end{equation}
For every $n$ choose $0\le m^{(n)}\le n$, $\mu^{(n)}\in\partitions[m^{(n)}]^k,\nu^{(n)}\in\partitions[n-m^{(n)}]^k$ such that $c^{\lambda^{(n)}}_{\mu^{(n)}\nu^{(n)}}\neq 0$ and attaining the maximum in \eqref{eq:directsummainestimate}. Choose a subsequence $n_l$ such that $\frac{1}{n_l}m^{(n_l)}$, $\frac{1}{m^{(n_l)}}\mu^{(n_l)}$ and $\frac{1}{n_l-m^{(n_l)}}\nu^{(n_l)}$ converge, and let their limits be $q,\overline{\mu},\overline{\nu}$. $\LR{\bullet}^k$ is closed, therefore $(\overline{\lambda},\overline{\mu},\overline{\nu})\in\LR{q}^k$.
\begin{equation}
\begin{split}
\ratefunction{\overline{\lambda}}{\psi\oplus\varphi}
 & = \lim_{n\to\infty}-\frac{1}{n}\log\norm{P^{\mathcal{H}\oplus\mathcal{K}}_{\lambda^{(n)}}(\psi\oplus\varphi)^{\otimes n}}^2  \\
 & = \lim_{l\to\infty}-\frac{1}{n_l}\log\norm{P^{\mathcal{H}\oplus\mathcal{K}}_{\lambda^{(n_l)}}(\psi\oplus\varphi)^{\otimes n_l}}^2  \\
 & \ge \liminf_{l\to\infty}-\frac{1}{n_l}\log \binom{n_l}{m^{(n_l)}}-\frac{m^{(n_l)}}{n_l}\frac{1}{m^{(n_l)}}\log \norm{P^{\mathcal{H}}_{\mu^{(n_l)}}\psi^{\otimes m^{(n_l)}}}^2  \\  &\qquad-\frac{n_l-m^{(n_l)}}{n_l}\frac{1}{n_l-m^{(n_l)}}\log \norm{P^{\mathcal{K}}_{\nu^{(n_l)}}\varphi^{\otimes n_l-m^{(n_l)}}}^2  \\
 & \ge -h(q)+q\ratefunction{\overline{\mu}}{\psi}+(1-q)\ratefunction{\overline{\nu}}{\varphi}.
\end{split}
\end{equation}
The first inequality follows from \eqref{eq:directsummainestimate} and the second one from Lemma~\ref{lem:subsequenceliminfbound}.
\end{proof}

\begin{proposition}[Rate function and local projections]\label{prop:rateprojection}
Let $\psi\in\mathcal{H}=\mathcal{H}_1\otimes\cdots\otimes\mathcal{H}_k$ and $\Pi=\Pi^2=\Pi^*\in\boundeds(\mathcal{H}_j)$ for some $j\in[k]$. Consider the vectors $\psi_1=(I\otimes\cdots\otimes I\otimes \Pi\otimes I\otimes\cdots I)\psi$ and $\psi_2=(I\otimes\cdots\otimes I\otimes(I-\Pi)\otimes I\otimes\cdots I)\psi$. For every $\overline{\mu},\overline{\nu}\in\overline{\partitions}^k$ and $q\in[0,1]$ the inequality
\begin{equation}
q\ratefunction{\overline{\mu}}{\psi_1}+(1-q)\ratefunction{\overline{\nu}}{\psi_2}-h(q) \ge \inf_{\substack{\overline{\lambda}\in\overline{\partitions}  \\  (\overline{\lambda},\overline{\mu},\overline{\nu})\in\LR{q}^k}}\ratefunction{\overline{\lambda}}{\psi}
\end{equation}
holds.
\end{proposition}
\begin{proof}
If $q=1$ (or $q=0$) then the right hand side is $\ratefunction{\overline{\mu}}{\psi}$ (or $\ratefunction{\overline{\nu}}{\psi}$) and the inequality follows from part~\ref{it:ratesimplemonotone} of Proposition~\ref{prop:ratebasic}, therefore we can assume $q\in(0,1)$ and that $\psi_1$ and $\psi_2$ are nonzero. Let $d=\sum_{j=1}^k\dim\mathcal{H}_j$, $0\le m\le n$, $\mu\in\partitions[m]^k$, $\nu\in\partitions[n-m]^k$. We define the disjoint projections $\Pi_0,\Pi_1,\ldots,\Pi_n$ via the generating function
\begin{equation}
\sum_{i=0}^n\Pi_it^i=I\otimes\cdots\otimes I\otimes ((I-\Pi)+t\Pi)^{\otimes n}\otimes I\otimes\cdots\otimes I.
\end{equation}

Consider the $S_n$-action on $\mathcal{H}^{\otimes n}$ that permutes the factors and permutations of the vector $(P_\mu\psi_1^{\otimes m})\otimes(P_\nu\psi_2^{\otimes n-m})$. $\Pi_m$ is $S_n$-invariant, therefore it commutes with $P_\lambda$. $\psi_1$ and $\psi_2$ are supported in the orthogonal subspaces $\Pi\mathcal{H}_j$ and $(I-\Pi)\mathcal{H}_j$ at site $j$ and $P_\mu\otimes P_\nu$ commutes both with $S_m\times S_{n-m}$ and with $\Pi^{\otimes m}\otimes(I-\Pi)^{n-m}$. Therefore every element $\sigma\in S_n$ either fixes this product (this happens iff $\sigma\in S_m\times S_{n-m}$) or sends it to an orthogonal one. The same is true for $\psi_1^{\otimes m}\otimes \psi_2^{\otimes n-m}$, and the sum of its distinct permutations is $\Pi_m\psi^{\otimes n}$. Using this and the vanishing condition \eqref{eq:LRvanishingproduct} we get (summing over distinct permutations $\sigma$, i.e. over a set of representatives of the cosets $S_n/S_m\times S_{n-m}$)
\begin{equation}\label{eq:projectionmainestimate}
\begin{split}
\binom{n}{m}\norm{P_{\mu}\psi_1^{\otimes m}}^2\norm{P_{\nu}\psi_2^{\otimes n-m}}^2
 & = \norm{\sum_{\sigma}\sigma\cdot(P_{\mu}\psi_1^{\otimes m}\otimes P_{\nu}\psi_2^{\otimes n-m})}^2  \\
 & = \sum_{\substack{\lambda\in\partitions[n]^k  \\  c^{\lambda}_{\mu\nu}\neq 0}}\norm{P_{\lambda}\sum_{\sigma}\sigma\cdot(P_{\mu}\psi_1^{\otimes m}\otimes P_{\nu}\psi_2^{\otimes n-m})}^2  \\
 & \le \sum_{\substack{\lambda\in\partitions[n]^k  \\  c^{\lambda}_{\mu\nu}\neq 0}}\norm{P_{\lambda}\sum_{\sigma}\sigma\cdot(\psi_1^{\otimes m}\otimes \psi_2^{\otimes n-m})}^2  \\
 & = \sum_{\substack{\lambda\in\partitions[n]^k  \\  c^{\lambda}_{\mu\nu}\neq 0}}\norm{P_{\lambda}\Pi_m\psi^{\otimes n}}^2  \\
 & \le \sum_{\substack{\lambda\in\partitions[n]^k  \\  c^{\lambda}_{\mu\nu}\neq 0}}\norm{P_{\lambda}\psi^{\otimes n}}^2  \\
 & \le (n+1)^d\max_{\substack{\lambda\in\partitions[n]^k  \\  c^{\lambda}_{\mu\nu}\neq 0}}\norm{P_{\lambda}\psi^{\otimes n}}^2
\end{split}
\end{equation}

Let $m^{(1)},m^{(2)},\ldots$, $\mu^{(1)},\mu^{(2)},\ldots$, $\nu^{(1)},\nu^{(2)},\ldots$ be sequences such that $1\le m^{(n)}\le n-1$, $\mu^{(n)}\in\partitions[m^{(n)}]^k$, $\nu^{(n)}\in\partitions[n-m^{(n)}]^k$ and
\begin{gather}
\lim_{n\to\infty}\frac{m^{(n)}}{n} = q  \\
\lim_{n\to\infty}\frac{1}{m^{(n)}}\mu^{(n)} = \overline{\mu}  \\
\lim_{n\to\infty}-\frac{1}{m^{(n)}}\log\norm{P_{\mu^{(n)}}\psi_1^{\otimes m^{(n)}}}^2 = \ratefunction{\overline{\mu}}{\psi_1}  \\
\lim_{n\to\infty}\frac{1}{n-m^{(n)}}\nu^{(n)} = \overline{\nu}  \\
\lim_{n\to\infty}-\frac{1}{n-m^{(n)}}\log\norm{P_{\nu^{(n)}}\psi_2^{\otimes n-m^{(n)}}}^2 = \ratefunction{\overline{\mu}}{\psi_2}.
\end{gather}
For every $n$ choose $\lambda^{(n)}\in\partitions[n]^k$ such that $c^{\lambda^{(n)}}_{\mu^{(n)}\nu^{(n)}}\neq 0$ and attaining the maximum in \eqref{eq:projectionmainestimate}. Choose a subsequence $n_l$ such that $\frac{1}{n_l}\lambda^{(n_l)}$ converges and let its limit be $\overline{\lambda}$. $\LR{\bullet}^k$ is closed, therefore $(\overline{\lambda},\overline{\mu},\overline{\nu})\in\LR{q}^k$.
\begin{equation}
\begin{split}
q\ratefunction{\overline{\mu}}{\psi_1}+(1-q)\ratefunction{\overline{\nu}}{\psi_2}-h(q)
 & = \lim_{n\to\infty}-\frac{1}{n}\log\norm{P_{\mu^{(n)}}\psi_1^{\otimes m^{(n)}}}^2  \\  & \qquad-\frac{1}{n}\log\norm{P_{\nu^{(n)}}\psi_2^{\otimes n-m^{(n)}}}^2-\frac{1}{n}\log\binom{n}{m^{(n)}}  \\
 & = \lim_{l\to\infty}-\frac{1}{n_l}\log\norm{P_{\mu^{(n_l)}}\psi_1^{\otimes m^{(n_l)}}}^2  \\  & \qquad-\frac{1}{n_l}\log\norm{P_{\nu^{(n_l)}}\psi_2^{\otimes n_l-m^{(n_l)}}}^2-\frac{1}{n_l}\log\binom{n_l}{m^{(n_l)}}  \\
 & \ge \lim_{l\to\infty}-\frac{1}{n_l}\log\norm{P_{\lambda^{(n_l)}}\psi^{\otimes n_l}}^2  \\
 & \ge \ratefunction{\overline{\lambda}}{\psi}.
\end{split}
\end{equation}
The first inequality follows from \eqref{eq:projectionmainestimate} and the second one from Lemma~\ref{lem:subsequenceliminfbound}.
\end{proof}

\section{LOCC spectral points}\label{sec:spectralpoints}

In this section we construct our family of functionals on pure unnormalised states which are monotone under trace-nonincreasing local operations and classical communication, additive under the direct sum and multiplicative under the tensor product, i.e. are LOCC spectral points in the terminology of \cite{jensen2019asymptotic}. This famlily is parametrised by a number $\alpha\in[0,1]$ and a point $\theta$ in the simplex $\distributions([k])$. For $\alpha=0$ they reduce to the quantum functionals introduced in \cite{christandl2018universal}, while for $\alpha=1$ they collapse to a single function, the norm squared.

Like the quantum functionals of \cite{christandl2018universal} (and also the support functionals of Strassen \cite{strassen1991degeneration}), our functionals come in two flavours, an ``upper'' and a ``lower'' family. In Section~\ref{sec:upperfunctional} we define the upper version in terms of asymptotic representation theoretical data: the Shannon entropy, which measures the growth rate of the dimensions of the representations of the symmetric group and the rate function studied in Section~\ref{sec:estimation}, and prove that it is submultiplicative, subadditive and satisfies \eqref{eq:monotonicitycondition}. In Section~\ref{sec:lowerfunctional} we define the lower counterpart in terms of the SLOCC orbit and prove that it is supermultiplicative. When $\alpha=0$ both functionals reduce to the ones defined in \cite{christandl2018universal} and are equal to each other. Below we will assume that $\alpha>0$. In this case we do not know if the lower and the upper functionals coincide. Nevertheless, in Section~\ref{sec:comparing} we show that the regularisation of the lower one equals the upper one and is superadditive, which implies that the upper functionals are LOCC spectral points.

\subsection{Upper functionals}\label{sec:upperfunctional}

We start by defining the ``upper'' family of the functionals. After the definition we prove its algebraic (Proposition~\ref{prop:uppersub}) and monotonicity (Proposition~\ref{prop:uppermonotone}) properties. The main tools in this section are the inequalities satisfied by the rate function, proved in Section~\ref{sec:estimation}, and the entropy inequalities related to the Kronecker and Littlewood--Richardson coefficients, as explained in Section~\ref{sec:notations}.
\begin{definition}[Upper functionals]\label{def:upper}
Let $\alpha\in(0,1]$ and $\theta\in\distributions([k])$. The logarithmic upper functional is
\begin{equation}\label{eq:upperdef}
\logupperLOCC{\alpha,\theta}(\ket{\psi}) = \sup_{\overline{\lambda}\in\overline{\partitions}^k}\left[(1-\alpha)\entropy_{\theta}(\overline{\lambda})-\alpha \ratefunction{\overline{\lambda}}{\psi}\right]
\end{equation}
and the upper functional is $\upperLOCC{\alpha,\theta}(\ket{\psi})=2^{\logupperLOCC{\alpha,\theta}(\ket{\psi})}$.
\end{definition}
We remark that $\ratefunction{\overline{\lambda}}{\psi}$ is lower semicontinuous and infinite outside a compact set, therefore for every $\alpha\in(0,1]$ the supremum in \eqref{eq:upperdef} is attained.

\begin{proposition}[Submultiplicativity and subadditivity of the upper functional]\label{prop:uppersub}
For any $\alpha\in(0,1]$, $p\in[0,\infty)$, $\theta\in\distributions([k])$ and vectors $\ket{\psi}\in\mathcal{H}_1\otimes\cdots\otimes\mathcal{H}_k$ and $\ket{\varphi}\in\mathcal{K}_1\otimes\cdots\otimes\mathcal{K}_k$ the followings hold:
\begin{enumerate}[(i)]
\item\label{it:upperscaling} $\upperLOCC{\alpha,\theta}(\sqrt{p}\ket{\psi})=p^\alpha \upperLOCC{\alpha,\theta}(\ket{\psi})$.
\item\label{it:uppersubmultiplicative} $\upperLOCC{\alpha,\theta}(\ket{\psi}\otimes\ket{\varphi})\le \upperLOCC{\alpha,\theta}(\ket{\psi})\upperLOCC{\alpha,\theta}(\ket{\varphi})$
\item\label{it:uppersubadditive} $\upperLOCC{\alpha,\theta}(\ket{\psi}\oplus\ket{\varphi})\le \upperLOCC{\alpha,\theta}(\ket{\psi})+\upperLOCC{\alpha,\theta}(\ket{\varphi})$
\end{enumerate}
\end{proposition}
\begin{proof}
\ref{it:upperscaling}:
The first term inside the supremum in \eqref{eq:upperdef} does not change if we replace $\ket{\psi}$ with $\sqrt{p}\ket{\psi}$, while the second term acquires an $\alpha\log p$ term by the scaling property of the rate function (Proposition~\ref{prop:ratebasic},~\ref{it:ratescaling}). Therefore $\logupperLOCC{\alpha,\theta}(\sqrt{p}\ket{\psi})=\logupperLOCC{\alpha,\theta}(\ket{\psi})+\alpha\log p$, which is equivalent to the statement.

\ref{it:uppersubmultiplicative}:
We use Proposition~\ref{prop:rateproduct} to bound $\logupperLOCC{\alpha,\theta}(\ket{\psi}\otimes\ket{\varphi})$ as
\begin{equation}
\begin{split}
\logupperLOCC{\alpha,\theta}(\ket{\psi}\otimes\ket{\varphi})
 & = \sup_{\overline{\lambda}\in\overline{\partitions}^k}\left[(1-\alpha)\entropy_{\theta}(\overline{\lambda})-\alpha \ratefunction{\overline{\lambda}}{\psi\otimes\varphi}\right]  \\
 & \le \sup_{\overline{\lambda}\in\overline{\partitions}^k}\sup_{\substack{\overline{\mu},\overline{\nu}\in\overline{\partitions}^k  \\  (\overline{\lambda},\overline{\mu},\overline{\nu})\in\Kron^k}}\left[(1-\alpha)\entropy_{\theta}(\overline{\lambda})-\alpha \ratefunction{\overline{\mu}}{\psi}-\alpha \ratefunction{\overline{\nu}}{\varphi}\right]  \\
 & \le \sup_{\overline{\lambda}\in\overline{\partitions}^k}\sup_{\substack{\overline{\mu},\overline{\nu}\in\overline{\partitions}^k  \\  (\overline{\lambda},\overline{\mu},\overline{\nu})\in\Kron^k}}\left[(1-\alpha)\left(\entropy_{\theta}(\overline{\mu})+\entropy_{\theta}(\overline{\nu})\right)-\alpha \ratefunction{\overline{\mu}}{\psi}-\alpha \ratefunction{\overline{\nu}}{\varphi}\right]  \\
 & = \sup_{\overline{\mu},\overline{\nu}\in\overline{\partitions}^k}\left[(1-\alpha)\left(\entropy_{\theta}(\overline{\mu})+\entropy_{\theta}(\overline{\nu})\right)-\alpha \ratefunction{\overline{\mu}}{\psi}-\alpha \ratefunction{\overline{\nu}}{\varphi}\right]  \\
 & = \logupperLOCC{\alpha,\theta}(\ket{\psi})+\logupperLOCC{\alpha,\theta}(\ket{\varphi}),
\end{split}
\end{equation}
where the second inequality uses that $(\overline{\lambda},\overline{\mu},\overline{\nu})\in\Kron^k$ implies $\entropy_{\theta}(\overline{\lambda})\le\entropy_{\theta}(\overline{\mu})+\entropy_{\theta}(\overline{\nu})$.

\ref{it:uppersubadditive}:
We use Proposition~\ref{prop:ratesum} to bound $\logupperLOCC{\alpha,\theta}(\ket{\psi}\oplus\ket{\varphi})$ as
\begin{equation}
\begin{split}
\logupperLOCC{\alpha,\theta}(\ket{\psi}\oplus\ket{\varphi})
 & = \sup_{\overline{\lambda}\in\overline{\partitions}^k}\left[(1-\alpha)\entropy_{\theta}(\overline{\lambda})-\alpha \ratefunction{\overline{\lambda}}{\psi\oplus\varphi}\right]  \\
 & \le \sup_{\overline{\lambda}\in\overline{\partitions}^k}\max_{q\in[0,1]}\sup_{\substack{\overline{\mu},\overline{\nu}\in\overline{\partitions}^k  \\  (\overline{\lambda},\overline{\mu},\overline{\nu})\in\LR{q}^k}}\big[(1-\alpha)\entropy_{\theta}(\overline{\lambda})  \\ &\qquad-\alpha(q\ratefunction{\overline{\mu}}{\psi}+(1-q)\ratefunction{\overline{\nu}}{\varphi}-h(q))\big]  \\
 & \le \sup_{\overline{\lambda}\in\overline{\partitions}^k}\max_{q\in[0,1]}\sup_{\substack{\overline{\mu},\overline{\nu}\in\overline{\partitions}^k  \\  (\overline{\lambda},\overline{\mu},\overline{\nu})\in\LR{q}^k}}\big[(1-\alpha)\left(q\entropy_{\theta}(\overline{\mu})+(1-q)\entropy_{\theta}(\overline{\nu})+h(q)\right)  \\ &\qquad-\alpha(q\ratefunction{\overline{\mu}}{\psi}+(1-q)\ratefunction{\overline{\nu}}{\varphi}-h(q))\big]  \\
 & = \max_{q\in[0,1]}\sup_{\overline{\mu},\overline{\nu}\in\overline{\partitions}^k}\big[(1-\alpha)\left(q\entropy_{\theta}(\overline{\mu})+(1-q)\entropy_{\theta}(\overline{\nu})+h(q)\right)  \\ &\qquad-\alpha(q\ratefunction{\overline{\mu}}{\psi}+(1-q)\ratefunction{\overline{\nu}}{\varphi}-h(q))\big]  \\
 & = \max_{q\in[0,1]}q\logupperLOCC{\alpha,\theta}(\ket{\psi})+(1-q)\logupperLOCC{\alpha,\theta}(\ket{\varphi})+h(q)  \\
 & = \log\left(\upperLOCC{\alpha,\theta}(\ket{\psi})+\upperLOCC{\alpha,\theta}(\ket{\varphi})\right).
\end{split}
\end{equation}
The second inequality uses that $(\overline{\lambda},\overline{\mu},\overline{\nu})\in\LR{q}^k$ implies $\entropy_{\theta}(\overline{\lambda})\le q\entropy_{\theta}(\overline{\mu})+(1-q)\entropy_{\theta}(\overline{\nu})+h(q)$.
\end{proof}

\begin{proposition}\label{prop:uppermonotone}
If $\Pi$ is a projection on $\mathcal{H}_j$ and $\ket{\psi}\in\mathcal{H}_1,\otimes\cdots\otimes\mathcal{H}_k$, then
\begin{equation}
\upperLOCC{\alpha,\theta}(\ket{\psi})^{1/\alpha}\ge \upperLOCC{\alpha,\theta}(\Pi_j\ket{\psi})^{1/\alpha}+\upperLOCC{\alpha,\theta}((I-\Pi)_j\ket{\psi})^{1/\alpha}.
\end{equation}
\end{proposition}
\begin{proof}
Let $\ket{\psi_1}=\Pi_j\ket{\psi}$ and $\ket{\psi_2}=(I-\Pi)_j\ket{\psi}$. Choose $\overline{\mu},\overline{\nu}\in\overline{\partitions}^k$ and $q\in[0,1]$. By Proposition~\ref{prop:rateprojection}, for every $\epsilon>0$, there is a $k$-tuple $\overline{\lambda}\in\overline{\partitions}^k$ such that
\begin{equation}
\ratefunction{\overline{\lambda}}{\psi}-\epsilon\le q\ratefunction{\overline{\mu}}{\psi_1}+(1-q)\ratefunction{\overline{\nu}}{\psi_2}-h(q)
\end{equation}
and $(\overline{\lambda},\overline{\mu},\overline{\nu})\in\LR{q}^k$.

By definition,
\begin{equation}
\begin{split}
\log \upperLOCC{\alpha,\theta}(\ket{\psi})^{1/\alpha}
 & = \frac{1}{\alpha}\logupperLOCC{\alpha,\theta}(\ket{\psi})  \\
 & \ge \frac{1-\alpha}{\alpha}\entropy_{\theta}(\overline{\lambda})-\ratefunction{\overline{\lambda}}{\psi}  \\
 & \ge \frac{1-\alpha}{\alpha}\left(q\entropy_{\theta}(\overline{\mu})+(1-q)\entropy_{\theta}(\overline{\nu})\right)-q\ratefunction{\overline{\mu}}{\psi_1}-(1-q)\ratefunction{\overline{\nu}}{\psi_2}+h(q)-\epsilon.
\end{split}
\end{equation}
The inequality also holds if we take $\epsilon\to 0$ and the supremum over $\overline{\mu},\overline{\nu}$, therefore
\begin{equation}
\log \upperLOCC{\alpha,\theta}(\ket{\psi})^{1/\alpha}\ge q\log \upperLOCC{\alpha,\theta}(\ket{\psi_1})^{1/\alpha}+(1-q)\log \upperLOCC{\alpha,\theta}(\ket{\psi_2})^{1/\alpha}+h(q).
\end{equation}
Finally, the maximum of the right hand side over $q$ is
\begin{equation}
\log\left(\upperLOCC{\alpha,\theta}(\Pi_j\ket{\psi})^{1/\alpha}+\upperLOCC{\alpha,\theta}((I-\Pi)_j\ket{\psi})^{1/\alpha}\right).
\end{equation}
\end{proof}

\subsection{Lower functionals}\label{sec:lowerfunctional}

Next we define the lower versions of our functionals and show basic properties such as monotonicity and supermultiplicativity (Proposition~\ref{prop:lowerbasic}). The latter property implies that the regularisation of the lower functional exists. The regularised functional inherits the algebraic properties and in addition satisfies superadditivity under the direct sum (Proposition~\ref{prop:asymptoticlowerbasic}).

In the following we will use the notation $\psi\succ\varphi$ to mean that there exist linear operators $A_1,\ldots,A_k$ ($A_j\in\boundeds(\mathcal{H}_j)$) satisfying $A_j^*A_j\le I$ for all $j\in[k]$ such that $\varphi=(A_1\otimes\cdots\otimes A_k)\psi$. In addition, when $\varphi\in\mathcal{H}\setminus\{0\}$ we set
\begin{equation}
\entropy_{\theta}(\varphi)=\sum_{j=1}^k\theta_j
\entropy\left(\frac{\ketbra{\varphi}{\varphi}_j}{\norm{\varphi}^2}\right),
\end{equation}
where the subscript $j$ refers to the $j$th marginal.
\begin{definition}[Lower functionals]\label{def:lower}
The logarithmic lower functional is 
\begin{equation}
\loglowerLOCC{\alpha,\theta}(\ket{\psi}) = \sup_{\substack{\ket{\varphi}\in\mathcal{H}  \\  \psi\succ\varphi}}\left[(1-\alpha)\entropy_{\theta}(\varphi)+\alpha\log\norm{\varphi}^2\right].
\end{equation}
The lower functional is $\lowerLOCC{\alpha,\theta}(\ket{\psi}):=2^{\loglowerLOCC{\alpha,\theta}(\ket{\psi})}$.
\end{definition}

\begin{proposition}[Basic properties of the lower functional]\label{prop:lowerbasic}
For any $\alpha\in[0,\infty)$, $p\in[0,\infty)$, $\theta\in\distributions([k])$ and vectors $\ket{\psi}\in\mathcal{H}_1\otimes\cdots\otimes\mathcal{H}_k$ and $\ket{\varphi}\in\mathcal{K}_1\otimes\cdots\otimes\mathcal{K}_k$ the followings hold:
\begin{enumerate}[(i)]
\item\label{it:lowerscaling} $\lowerLOCC{\alpha,\theta}(\sqrt{p}\ket{\psi})=p^\alpha \lowerLOCC{\alpha,\theta}(\ket{\psi})$.
\item\label{it:lowernormalised} $\lowerLOCC{\alpha,\theta}(\unittensor{r})=r$.
\item\label{it:lowersimplemonotone} If $\psi_1\succ\psi_2$ then $\lowerLOCC{\alpha,\theta}(\ket{\psi_1})\ge\lowerLOCC{\alpha,\theta}(\ket{\psi_2})$.
\item\label{it:lowersupermultiplicative} $\lowerLOCC{\alpha,\theta}(\ket{\psi_1}\otimes\ket{\psi_2})\ge \lowerLOCC{\alpha,\theta}(\ket{\psi_1})\lowerLOCC{\alpha,\theta}(\ket{\psi_2})$
\end{enumerate}
\end{proposition}
\begin{proof}
\ref{it:lowerscaling}:
If we replace $\ket{\psi}$ with $\sqrt{p}\ket{\psi}$, then the allowed $\ket{\varphi}$ also get rescaled by $\sqrt{p}$. The first term in the supremum is not sensitive to this, while the second gets an additional $\alpha\log p$ term. Therefore $\loglowerLOCC{\alpha,\theta}(\sqrt{p}\ket{\psi})=\loglowerLOCC{\alpha,\theta}(\ket{\psi})+\alpha\log p$, which is equivalent to the statement.

\ref{it:lowernormalised}:
The entropies in the supremum are upper bounded by $\log r$ for any choice of $\ket{\varphi}$ since the local ranks cannot increase under separable operations. The norm is also nonincreasing, therefore $\log\norm{\varphi}\le\log\norm{\unittensor{r}}^2=\log r$. This proves that $\loglowerLOCC{\alpha,\theta}(\unittensor{r})\le \log r$. On the other hand, $\ket{\varphi}=\unittensor{r}$ is feasible and achieves this upper bound.

\ref{it:lowersimplemonotone}:
If $\psi_2\succ\varphi$ then by assumption also $\psi_1\succ\varphi$, therefore the supremum for $\psi_1$ is taken over a larger set than for $\psi_2$, which implies $\loglowerLOCC{\alpha,\theta}(\ket{\psi_1})\ge\loglowerLOCC{\alpha,\theta}(\psi_2)$.

\ref{it:lowersupermultiplicative}:
If $\psi_i\succ\varphi_i$ ($i=1,2$) then $\psi_1\otimes\psi_2\succ\varphi_1\otimes\varphi_2$, therefore
\begin{equation}
\begin{split}
\loglowerLOCC{\alpha,\theta}(\ket{\psi_1}\otimes\ket{\psi_2})
 & \ge (1-\alpha)\entropy_{\theta}(\varphi_1\otimes\varphi_2)+\alpha\log\norm{\varphi_1\otimes\varphi_2}^2  \\
 & = (1-\alpha)\entropy_{\theta}(\varphi_1)+\alpha\log\norm{\varphi_1}^2+(1-\alpha)\entropy_{\theta}(\varphi_2)+\alpha\log\norm{\varphi_2}^2.
\end{split}
\end{equation}
Now take the supremum over the admissible $\ket{\varphi_1}$ and $\ket{\varphi_2}$ to get $\loglowerLOCC{\alpha,\theta}(\ket{\psi_1}\otimes\ket{\psi_2})\ge \loglowerLOCC{\alpha,\theta}(\ket{\psi_1})+\loglowerLOCC{\alpha,\theta}(\ket{\psi_2})$.
\end{proof}

Part~\ref{it:lowersupermultiplicative} of Proposition~\ref{prop:lowerbasic} says that $\loglowerLOCC{\alpha,\theta}$ is superadditive under tensor product. We also have the additive upper bound $\loglowerLOCC{\alpha,\theta}(\ket{\psi})\le(1-\alpha)\sum_{j=1}^k\theta_j\log\dim\mathcal{H}_j+\alpha\log\norm{\psi}^2$. Therefore the regularised lower functional exists and is finite, so the following definition is meaningful.
\begin{definition}[Asymptotic lower functional]\label{def:asymptoticlower}
The asymptotic logarithmic lower functional is
\begin{equation}
\logasymptoticlowerLOCC{\alpha,\theta}(\ket{\psi})=\lim_{n\to\infty}\frac{1}{n}\loglowerLOCC{\alpha,\theta}(\ket{\psi}^{\otimes n})
\end{equation}
and the asymptotic lower functional is
\begin{equation}
\asymptoticlowerLOCC{\alpha,\theta}(\ket{\psi})=\lim_{n\to\infty}\sqrt[n]{\lowerLOCC{\alpha,\theta}(\ket{\psi}^{\otimes n})}=2^{\logasymptoticlowerLOCC{\alpha,\theta}(\ket{\psi})}.
\end{equation}
\end{definition}

\begin{proposition}[Basic properties of the asymptotic lower functional]\label{prop:asymptoticlowerbasic}
For any $\alpha\in[0,\infty)$, $p\in[0,\infty)$, $\theta\in\distributions([k])$ and vectors $\ket{\psi}\in\mathcal{H}_1\otimes\cdots\otimes\mathcal{H}_k$ and $\ket{\varphi}\in\mathcal{K}_1\otimes\cdots\otimes\mathcal{K}_k$ the followings hold:
\begin{enumerate}[(i)]
\item\label{it:asymptoticlowerscaling} $\asymptoticlowerLOCC{\alpha,\theta}(\sqrt{p}\ket{\psi})=p^\alpha \asymptoticlowerLOCC{\alpha,\theta}(\ket{\psi})$.
\item\label{it:asymptoticlowernormalised} $\asymptoticlowerLOCC{\alpha,\theta}(\unittensor{r})=r$.
\item\label{it:asymptoticlowersupermultiplicative} $\asymptoticlowerLOCC{\alpha,\theta}(\ket{\psi_1}\otimes\ket{\psi_2})\ge \asymptoticlowerLOCC{\alpha,\theta}(\ket{\psi_1})\asymptoticlowerLOCC{\alpha,\theta}(\ket{\psi_2})$
\item\label{it:asymptoticlowersuperadditive} $\asymptoticlowerLOCC{\alpha,\theta}(\ket{\psi_1}\oplus\ket{\psi_2})\ge \asymptoticlowerLOCC{\alpha,\theta}(\ket{\psi_1})+\asymptoticlowerLOCC{\alpha,\theta}(\ket{\psi_2})$
\end{enumerate}
\end{proposition}
\begin{proof}
\ref{it:asymptoticlowerscaling}:
By part~\ref{it:lowerscaling} of Proposition~\ref{prop:lowerbasic} we have
\begin{equation}
\begin{split}
\asymptoticlowerLOCC{\alpha,\theta}(\sqrt{p}\ket{\psi})
 & = \lim_{n\to\infty}\sqrt[n]{\lowerLOCC{\alpha,\theta}((\sqrt{p}\ket{\psi})^{\otimes n})}  \\
 & = \lim_{n\to\infty}\sqrt[n]{p^{n\alpha}\lowerLOCC{\alpha,\theta}(\ket{\psi})}=p^\alpha\asymptoticlowerLOCC{\alpha,\theta}(\sqrt{p}\ket{\psi}).
\end{split}
\end{equation}

\ref{it:asymptoticlowernormalised}:
By part~\ref{it:lowernormalised} of Proposition~\ref{prop:lowerbasic} we have
\begin{equation}
\asymptoticlowerLOCC{\alpha,\theta}(\unittensor{r})
 = \lim_{n\to\infty}\sqrt[n]{\lowerLOCC{\alpha,\theta}(\unittensor{r}^{\otimes n})}
 = \lim_{n\to\infty}\sqrt[n]{\lowerLOCC{\alpha,\theta}(\unittensor{r^n})}
 = \lim_{n\to\infty}\sqrt[n]{r^n}=r.
\end{equation}

\ref{it:asymptoticlowersupermultiplicative}:
By part~\ref{it:lowersupermultiplicative} of Proposition~\ref{prop:lowerbasic} we have
\begin{equation}
\begin{split}
\asymptoticlowerLOCC{\alpha,\theta}(\ket{\psi_1}\otimes\ket{\psi_2})
 & = \lim_{n\to\infty}\sqrt[n]{\lowerLOCC{\alpha,\theta}((\ket{\psi_1}\otimes\ket{\psi_2})^{\otimes n})}  \\
 & \ge \lim_{n\to\infty}\sqrt[n]{\lowerLOCC{\alpha,\theta}(\ket{\psi_1}^{\otimes n})\lowerLOCC{\alpha,\theta}(\ket{\psi_2}^{\otimes n})}  \\
 & = \asymptoticlowerLOCC{\alpha,\theta}(\ket{\psi_1})\asymptoticlowerLOCC{\alpha,\theta}(\ket{\psi_2})
\end{split}
\end{equation}

\ref{it:asymptoticlowersuperadditive}:
Let $q\in[0,1]$. For any $n\in\naturals$ we expand the tensor power of the direct sum as
\begin{equation}
(\ket{\psi_1}\oplus\ket{\psi_2})^{\otimes n}=\bigoplus_{m=0}^n\unittensor{\binom{n}{m}}\otimes\ket{\psi_1}^{\otimes m}\otimes\ket{\psi_2}^{\otimes(n-m)}.
\end{equation}
Keep only the $m=\lfloor qn\rfloor$ term. This can be done by a local projection, which is a separable operation with one Kraus operator, i.e. $(\ket{\psi_1}\oplus\ket{\psi_2})^{\otimes n}\succ\unittensor{\binom{n}{m}}\otimes\ket{\psi_1}^{\otimes m}\otimes\ket{\psi_2}^{\otimes(n-m)}$. Therefore, using \ref{it:lowersimplemonotone} and \ref{it:lowersupermultiplicative} of Proposition~\ref{prop:lowerbasic} we have
\begin{equation}
\begin{split}
\loglowerLOCC{\alpha,\theta}((\ket{\psi_1}\oplus\ket{\psi_2})^{\otimes n})
 & \ge \loglowerLOCC{\alpha,\theta}\left(\unittensor{\binom{n}{\lfloor qn\rfloor}}\otimes\ket{\psi_1}^{\otimes\lfloor qn\rfloor}\otimes\ket{\psi_2}^{\otimes n-\lfloor qn\rfloor}\right)  \\
 & \ge \log\binom{n}{\lfloor qn\rfloor}+\loglowerLOCC{\alpha,\theta}(\ket{\psi_1}^{\otimes\lfloor qn\rfloor})+\loglowerLOCC{\alpha,\theta}(\ket{\psi_2}^{\otimes n-\lfloor qn\rfloor}).
\end{split}
\end{equation}
Divide both sides by $n$ and take the limit $n\to\infty$:
\begin{equation}
\begin{split}
\logasymptoticlowerLOCC{\alpha,\theta}(\ket{\psi_1}\oplus\ket{\psi_2})
 & = \lim_{n\to\infty}\frac{1}{n}\loglowerLOCC{\alpha,\theta}((\ket{\psi_1}\oplus\ket{\psi_2})^{\otimes n})  \\
 & \ge \lim_{n\to\infty}\frac{1}{n}\log\binom{n}{\lfloor qn\rfloor}+\lim_{n\to\infty}\frac{\lfloor qn\rfloor}{n}\frac{1}{\lfloor qn\rfloor}\loglowerLOCC{\alpha,\theta}(\ket{\psi_1}^{\otimes\lfloor qn\rfloor})  \\ &\qquad+\lim_{n\to\infty}\frac{n-\lfloor qn\rfloor}{n}\frac{1}{n-\lfloor qn\rfloor}\loglowerLOCC{\alpha,\theta}(\ket{\psi_2}^{\otimes n-\lfloor qn\rfloor})  \\
 & = h(q)+q\logasymptoticlowerLOCC{\alpha,\theta}(\ket{\psi_1})+(1-q)\logasymptoticlowerLOCC{\alpha,\theta}(\ket{\psi_2}).
\end{split}
\end{equation}
To finish the proof, take the maximum of the right hand side over $q$.
\end{proof}

\subsection{Comparing the functionals}\label{sec:comparing}

The aim of this section is to prove that $\logasymptoticlowerLOCC{\alpha,\theta}=\logupperLOCC{\alpha,\theta}$, from which our main theorem immediately follows. The inequality $\logasymptoticlowerLOCC{\alpha,\theta}(\ket{\psi})\le \logupperLOCC{\alpha,\theta}(\ket{\psi})$ relies on the spectrum estimation theorem (similarly to \cite[Theorem 3.24]{christandl2018universal}) and basic properties of the rate function. The reverse inequality is an application of the dimension estimates \eqref{eq:Schurdimension} and \eqref{eq:symmetricgrpdimension}.

\begin{proposition}
$\loglowerLOCC{\alpha,\theta}(\ket{\psi})\le \logupperLOCC{\alpha,\theta}(\ket{\psi})$
\end{proposition}
\begin{proof}
Let $\ket{\varphi}=(A_1\otimes\cdots\otimes A_k)\ket{\psi}$ where $A_j^*A_j\le I$ for all $j$. Let $\overline{\lambda}=(\overline{\lambda_1},\ldots,\overline{\lambda_k})\in\overline{\partitions}^k$ be the decreasingly ordered marginal spectra of $\ketbra{\varphi}{\varphi}/\norm{\varphi}^2$. Then we have
\begin{equation}
\begin{split}
\ratefunction{\overline{\lambda}}{\psi}
 & \le\ratefunction{\overline{\lambda}}{\varphi}  \\
 & = \ratefunction{\overline{\lambda}}{\varphi/\norm{\varphi}}+\log\norm{\varphi}^{-2}  \\
 & = -\log\norm{\varphi}^2
\end{split}
\end{equation}
by Proposition~\ref{prop:ratebasic}. It follows from Definition~\ref{def:upper} that
\begin{equation}
\begin{split}
\logupperLOCC{\alpha,\theta}(\ket{\psi})
 & \ge (1-\alpha)\entropy_{\theta}(\overline{\lambda})-\alpha \ratefunction{\overline{\lambda}}{\psi}  \\
 & \ge (1-\alpha)\entropy_{\theta}(\overline{\lambda})+\alpha\log\norm{\varphi}^2  \\
 & = (1-\alpha)\entropy_{\theta}(\varphi)+\alpha\log\norm{\varphi}^2.
\end{split}
\end{equation}
The supremum of the right hand side over the possible $\ket{\varphi}$ is $\loglowerLOCC{\alpha,\theta}(\ket{\psi})$.
\end{proof}

\begin{corollary}\label{cor:asymptoticlowerleupper}
$\logasymptoticlowerLOCC{\alpha,\theta}(\ket{\psi})\le \logupperLOCC{\alpha,\theta}(\ket{\psi})$.
\end{corollary}
\begin{proof}
$\logupperLOCC{\alpha,\theta}$ is subadditive under tensor product, therefore
\begin{equation}
\logasymptoticlowerLOCC{\alpha,\theta}(\ket{\psi})=\lim_{n\to\infty}\frac{1}{n}\loglowerLOCC{\alpha,\theta}(\ket{\psi}^{\otimes n})\le\lim_{n\to\infty}\frac{1}{n}\logupperLOCC{\alpha,\theta}(\ket{\psi}^{\otimes n})\le\logupperLOCC{\alpha,\theta}(\ket{\psi}).
\end{equation}
\end{proof}

\begin{proposition}\label{prop:asymptoticlowergeupper}
$\logasymptoticlowerLOCC{\alpha,\theta}(\ket{\psi})\ge \logupperLOCC{\alpha,\theta}(\ket{\psi})$.
\end{proposition}
\begin{proof}
Let $\psi\in\mathcal{H}=\mathcal{H}_1\otimes\cdots\otimes\mathcal{H}_k$, choose $\overline{\lambda}\in\overline{\partitions}^k$. Let $(\lambda^{(n)})_{n\in\naturals}$ be a sequence such that $\lambda^{(n)}\in\partitions[n]^k$, $\frac{1}{n}\lambda^{(n)}\to\overline{\lambda}$ and
\begin{equation}
\lim_{n\to\infty}-\frac{1}{n}\log\norm{P_{\lambda^{(n)}}\psi^{\otimes n}}^2=\ratefunction{\overline{\lambda}}{\psi},
\end{equation}
as in Lemma~\ref{lem:ratefunctionsinglelambdalimit}. By Definition~\ref{def:lower}, we can estimate $\loglowerLOCC{\alpha,\theta}(\ket{\psi}^{\otimes n})$ using $\psi^{\otimes n}\succ P_{\lambda^{(n)}}\psi^{\otimes n}=:\varphi_n$.

The $j$th marginal of $\varphi_n/\norm{\varphi_n}$ is a state on $\mathbb{S}_{\lambda^{(n)}_j}(\mathcal{H}_j)\otimes[\lambda^{(n)}_j]$, invariant under the action of $S_n$ on the second factor, thus the reduced state on that factor is maximally mixed. We use the triangle inequality for the von Neumann entropy and the dimension estimates \eqref{eq:Schurdimension} and \eqref{eq:symmetricgrpdimension} to get
\begin{equation}
\begin{split}
\entropy_{\theta}(\varphi_n)
 & \ge \sum_{j=1}^k\theta_j\left(\log\dim[\lambda^{(n)}_j]-\log\dim\mathbb{S}_{\lambda^{(n)}_j}(\mathcal{H}_j)\right)  \\
 & \ge n\sum_{j=1}^k\theta_j\left[\entropy\left(\frac{1}{n}\lambda^{(n)}_j\right)-\log(n+1)^{d_j(d_j-1)/2}(n+d_j)^{(d_j+2)(d_j-1)/2}\right]  \\
 & \ge n\entropy_{\theta}\left(\frac{1}{n}\lambda^{(n)}\right)-\log(n+d)^{d^2-1},
\end{split}
\end{equation}
where $d_j=\dim\mathcal{H}_j$ and $d=\dim\mathcal{H}$.

This leads to the lower bound
\begin{equation}
\begin{split}
\logasymptoticlowerLOCC{\alpha,\theta}(\ket{\psi})
 & \ge \limsup_{n\to\infty}\frac{1}{n}\loglowerLOCC{\alpha,\theta}(\ket{\psi}^{\otimes n})  \\
 & \ge \limsup_{n\to\infty}\frac{1}{n}\left[(1-\alpha)\entropy(\varphi_n)+\alpha\log\norm{\varphi_n}^2\right]  \\
 & \ge \limsup_{n\to\infty}\left[(1-\alpha)\entropy_{\theta}\left(\frac{1}{n}\lambda^{(n)}\right)-\frac{1}{n}(1-\alpha)\log(n+d)^{d^2-1}+\alpha\frac{1}{n}\log\norm{\varphi_n}^2\right]  \\
 & = (1-\alpha)\entropy_{\theta}(\overline{\lambda})-\alpha\ratefunction{\overline{\lambda}}{\psi}.
\end{split}
\end{equation}
To finish the proof we take the supremum over $\overline{\lambda}$ so that the right hand side becomes $\logupperLOCC{\alpha,\theta}(\ket{\psi})$.
\end{proof}

\begin{proof}[Proof of Theorem~\ref{thm:main}]
By Corollary~\ref{cor:asymptoticlowerleupper} and Proposition~\ref{prop:asymptoticlowergeupper} we see that $\logasymptoticlowerLOCC{\alpha,\theta}=\logupperLOCC{\alpha,\theta}$. The scaling property \ref{it:mainscaling} follows from part~\ref{it:upperscaling} of Proposition~\ref{prop:uppersub} (or from part~\ref{it:asymptoticlowerscaling} of Proposition~\ref{prop:asymptoticlowerbasic}), the normalisation \ref{it:mainnormalised} is the same as part~\ref{it:asymptoticlowernormalised} of Proposition~\ref{prop:asymptoticlowerbasic}. The multiplicativity \ref{it:mainmultiplicative} follows from part~\ref{it:asymptoticlowersupermultiplicative} of Proposition~\ref{prop:asymptoticlowerbasic} and part~\ref{it:uppersubmultiplicative} of Proposition~\ref{prop:uppersub}. The additivity \ref{it:mainadditive} follows from part~\ref{it:asymptoticlowersuperadditive} of Proposition~\ref{prop:asymptoticlowerbasic} and part~\ref{it:uppersubadditive} of Proposition~\ref{prop:uppersub}. The monotonicity property~\ref{it:mainmonotone} is a consequence of the scaling property, Proposition~\ref{prop:uppermonotone} and the equivalent condition \eqref{eq:monotonicitycondition} (from the characterisation \cite[Theorem 3.1.]{jensen2019asymptotic} of LOCC spectral points).
\end{proof}

\section*{Acknowledgements}

I thank Matthias Christandl for numerous discussions on the simultaneous spectrum estimation problem. This research was supported by the J\'anos Bolyai Research Scholarship of the Hungarian Academy of Sciences and the National Research, Development and Innovation Fund of Hungary within the Quantum Technology National Excellence Program (Project Nr.~2017-1.2.1-NKP-2017-00001) and via the research grants K124152, KH129601. Part of this work was done while the author was with QMATH.

\appendix

\section{Proofs of technical statements}\label{sec:technical}

\begin{proposition}\label{prop:limballexists}
Let $\overline{\lambda}\in\overline{\partitions}^k$, $\epsilon>0$ and $\varphi\in\mathcal{H}=\mathcal{H}_1\otimes\cdots\otimes\mathcal{H}_k$. Then the limit
\begin{equation}
\lim_{n\to\infty}\frac{1}{n}\log\sum_{\substack{\lambda\in\partitions[n]^k  \\  \frac{1}{n}\lambda\in \ball{\epsilon}{\overline{\lambda}}}}\norm{P_\lambda\varphi^{\otimes n}}^2
\end{equation}
exists (possibly $-\infty$, with $\log0$ interpreted as $-\infty$).
\end{proposition}
\begin{proof}
If $P_\lambda\varphi^{\otimes n}=0$ for every $n$ and $\lambda$ with $\frac{1}{n}\lambda\in \ball{\epsilon}{\overline{\lambda}}$ then the sequence is constant $-\infty$. We can thus assume $P_\lambda\varphi^{\otimes n}\neq 0$ for some $\lambda$. In particular, $\varphi\neq 0$.

Fix a maximal torus and a choice of positive Weyl chamber for $U(\mathcal{H}_1)\times\cdots\times U(\mathcal{H}_k)$. Let $W_\lambda\in\boundeds(\mathcal{H}^{\otimes n})$ be the orthogonal projection onto the subspace of highest weight vectors with weight $\lambda$. Then
\begin{equation}\label{eq:wandpestimate}
W_\lambda\le P_\lambda = \dim \mathbb{S}_\lambda(\mathcal{H})\int_{U(\mathcal{H}_1)\times\cdots\times U(\mathcal{H}_k)} (U^{\otimes n})^*W_\lambda U^{\otimes n}\ed U,
\end{equation}
where the integration is with respect to the Haar probability measure. $\dim \mathbb{S}_\lambda(\mathcal{H})\le(n+1)^{d_1}$ where $d_1=\sum_{j=1}^k\frac{\dim\mathcal{H}_j(\dim\mathcal{H}_j-1)}{2}$.

The only $k$-tuple for $n=1$ is $((1),\ldots,(1))$, therefore $P_{((1),\ldots,(1))}\varphi=\varphi\neq 0$. By \eqref{eq:wandpestimate} there is an $U_0\in U(\mathcal{H}_1)\times\cdots\times U(\mathcal{H}_k)$ such that $W_{((1),\ldots,(1))} U_0\varphi\neq 0$. In fact, the set
\begin{equation}\label{eq:nonvanishingU}
\setbuild{U\in U(\mathcal{H}_1)\times\cdots\times U(\mathcal{H}_k)}{W_{((1),\ldots,(1))} U\varphi\neq 0}
\end{equation}
is dense since for any $U$ one can find $A\in \mathfrak{u}(\mathcal{H}_1)\times\cdots\times \mathfrak{u}(\mathcal{H}_k)$ with $U_0U^{-1}=e^A$ and then the function
\begin{equation}
t\mapsto \norm{e^{tA}U\varphi}^2
\end{equation}
is real analytic and not identically zero, therefore arbitrarily small values of $t$ exist with $e^{tA}U\varphi\neq 0$.

Let $n_0\in\naturals$ such that
\begin{equation}\label{eq:sumnonzero}
\sum_{\substack{\lambda\in\partitions[n_0]^k  \\  \frac{1}{n_0}\lambda\in \ball{\epsilon}{\overline{\lambda}}}}\norm{P_\lambda\varphi^{\otimes n_0}}^2\neq 0.
\end{equation}
The number of nonzero terms is at most $(n_0+1)^{d_2}$ where $d_2=\sum_{j=1}^k\mathcal{H}_j$, therefore there is a $\lambda_0\in\partitions[n_0]^k$ such that $\frac{1}{n_0}\lambda_0\in \ball{\epsilon}{\overline{\lambda}}$ and
\begin{equation}
\norm{P_{\lambda_0}\varphi^{\otimes n_0}}^2\ge(n_0+1)^{-d_2}\sum_{\substack{\lambda\in\partitions[n_0]^k  \\  \frac{1}{n_0}\lambda\in \ball{\epsilon}{\overline{\lambda}}}}\norm{P_\lambda\varphi^{\otimes n_0}}^2.
\end{equation}
\eqref{eq:wandpestimate} implies that
\begin{equation}
\begin{split}
\norm{P_{\lambda_0}\varphi^{\otimes n_0}}^2
 & = \langle\varphi^{\otimes n_0},P_{\lambda_0}\varphi^{\otimes n_0}\rangle  \\
 & \le (n_0+1)^{d_1}\int_{U(\mathcal{H}_1)\times\cdots\times U(\mathcal{H}_k)} \langle (U\varphi)^{\otimes n_0},W_{\lambda_0}(U\varphi)^{\otimes n_0}\rangle\ed U,
\end{split}
\end{equation}
therefore there exists $U\in U(\mathcal{H}_1)\times\cdots\times U(\mathcal{H}_k)$ such that (with $d=d_1+d_2$)
\begin{equation}\label{eq:hwvlowerbound}
\norm{W_{\lambda_0}(U\varphi)^{\otimes n_0}}^2\ge(n_0+1)^{-d}\sum_{\substack{\lambda\in\partitions[n_0]^k  \\  \frac{1}{n_0}\lambda\in \ball{\epsilon}{\overline{\lambda}}}}\norm{P_\lambda\varphi^{\otimes n_0}}^2
\end{equation}
The left hand side of \eqref{eq:hwvlowerbound} is continuous in $U$ and the set \eqref{eq:nonvanishingU} is dense, therefore for every $\delta>0$ there exists a $U_0\in U(\mathcal{H}_1)\times\cdots\times U(\mathcal{H}_k)$ such that $W_{((1),\ldots,(1))}U_0\varphi\neq 0$ and
\begin{equation}
\norm{W_{\lambda_0}(U_0\varphi)^{\otimes n_0}}^2\ge(1-\delta)(n_0+1)^{-d}\sum_{\substack{\lambda\in\partitions[n_0]^k  \\  \frac{1}{n_0}\lambda\in \ball{\epsilon}{\overline{\lambda}}}}\norm{P_\lambda\varphi^{\otimes n_0}}^2.
\end{equation}

Let $n\in\naturals$ be arbitrary and write $n=qn_0+r$ with $q,r\in\naturals$, $r<n_0$. Then $\mu:=q\lambda_0+r((1),(1),\ldots,(1))\in\partitions[n]^k$ and
\begin{equation}
\frac{1}{n}\mu-\frac{1}{n_0}\lambda_0=\frac{r}{n}\left(((1),(1),\ldots,(1)))-\frac{1}{n_0}\lambda_0\right),
\end{equation}
which goes to $0$ as $n\to\infty$, therefore $\frac{1}{n}\mu\in\ball{\epsilon}{\overline{\lambda}}$ for every large enough $n$.

The tensor product of two highest weight vectors is itself a highest weight vector with the sum of the two weights, therefore $P_\mu\ge W_\mu\ge W_{\lambda_0}^{\otimes q}\otimes W_{((1),(1),\ldots,(1))}^{\otimes r}$. From this we get
\begin{equation}
\begin{split}
\liminf_{n\to\infty}\frac{1}{n}\log\sum_{\substack{\lambda\in\partitions[n]^k  \\  \frac{1}{n}\lambda\in \ball{\epsilon}{\overline{\lambda}}}}\norm{P_\lambda\varphi^{\otimes n}}^2
 & \ge \liminf_{n\to\infty}\frac{1}{n}\log\norm{W_{\lambda_0}(U_0\varphi)^{\otimes n_0}}^{2q}\norm{W_{((1),(1),\ldots,(1))}(U_0\varphi)}^{2r}  \\
 & = \liminf_{n\to\infty}\left(1-\frac{r}{n}\right)\frac{1}{n_0}\log\norm{W_{\lambda_0}(U_0\varphi)^{\otimes n_0}}^{2}+\frac{r}{n}\log\norm{W_{((1),(1),\ldots,(1))}(U_0\varphi)}^2  \\
 & = \frac{1}{n_0}\log\norm{W_{\lambda_0}(U_0\varphi)^{\otimes n_0}}^{2}  \\
 & \ge \frac{1}{n_0}\log(1-\delta)-\frac{d}{n_0}\log(n_0+1)+\frac{1}{n_0}\sum_{\substack{\lambda\in\partitions[n_0]^k  \\  \frac{1}{n_0}\lambda\in \ball{\epsilon}{\overline{\lambda}}}}\norm{P_\lambda\varphi^{\otimes n_0}}^2.
\end{split}
\end{equation}
The inequality is true for all $\delta>0$, therefore also for $\delta=0$. In particular, \eqref{eq:sumnonzero} is not only true for $n_0$, but also for every large enough $n$. Therefore we can take $\limsup_{n_0\to\infty}$ of both sides, and the resulting inequality means that the limit exists.
\end{proof}

\begin{proof}[Proof of Lemma~\ref{lem:subsequenceliminfbound}]
Let $\epsilon>0$. $\frac{n_l}{\lambda^{(n_l)}}\in\ball{\epsilon}{\overline{\lambda}}$ for every large enough $l$, therefore
\begin{equation}
\norm{P_{\lambda^{(n_l)}}}^2\le\sum_{\substack{\lambda\in\partitions[n]^k  \\  \frac{1}{n}\lambda\in \ball{\epsilon}{\overline{\lambda}}}}\norm{P_\lambda\varphi^{\otimes n}}^2.
\end{equation}
Apply the logarithm to both sides, multiply by $-\frac{1}{n_l}$ and take the limit inferior:
\begin{equation}
\begin{split}
\liminf_{n\to\infty}-\frac{1}{n_l}\norm{P_{\lambda^{(n_l)}}}^2
 & \ge\liminf_{n\to\infty}-\frac{1}{n_l}\sum_{\substack{\lambda\in\partitions[n_l]^k  \\  \frac{1}{n_l}\lambda\in \ball{\epsilon}{\overline{\lambda}}}}\norm{P_\lambda\varphi^{\otimes n_l}}^2  \\
 & =\lim_{n\to\infty}-\frac{1}{n}\sum_{\substack{\lambda\in\partitions[n]^k  \\  \frac{1}{n}\lambda\in \ball{\epsilon}{\overline{\lambda}}}}\norm{P_\lambda\varphi^{\otimes n}}^2,
\end{split}
\end{equation}
in the last step using Proposition~\ref{prop:limballexists}. The claim follows after taking the limit $\epsilon\to 0$.
\end{proof}

\begin{proof}[Proof of Lemma~\ref{lem:ratefunctionsinglelambdalimit}]
For every $\epsilon>0$ let
\begin{equation}
I_\epsilon=\lim_{n\to\infty}-\frac{1}{n}\log\sum_{\substack{\lambda\in\partitions[n]^k  \\  \frac{1}{n}\lambda\in \ball{\epsilon}{\overline{\lambda}}}}\norm{P_\lambda\varphi^{\otimes n}}^2
\end{equation}
so that $\lim_{\epsilon\to 0}I_\epsilon=\ratefunction{\overline{\lambda}}{\varphi}$.

Choose a sequence $(\epsilon_m)_{m=1}^\infty$ such that $\epsilon_1\ge 2$ (so that any $\epsilon_1$-ball in $\overline{\partitions}^k$ is equal to $\overline{\partitions}^k$), $\lim_{m\to\infty}\epsilon_m=0$ and for each $m$ a sequence $(\lambda^{(m,n)})_{n}$ such that $\lambda^{(m,n)}\in\partitions[n]^k$, $\frac{1}{n}\lambda^{(m,n)}\in\ball{\epsilon_m}{\overline{\lambda}}$ and $\norm{P_{\lambda^{(m,n)}}\varphi^{\otimes n}}$ is maximal subject to these conditions (such a choice may not be possible for small $n$, we only define the sequence for large enough $n$). Then
\begin{equation}
\lim_{n\to\infty}-\frac{1}{n}\log\norm{P_{\lambda^{(m,n)}}\varphi^{\otimes n}}^2=I_{\epsilon_m}.
\end{equation}
Let $n_0(m)$ be the smallest integer such that for all $n\ge n_0(m)$ the inequality
\begin{equation}
-\frac{1}{n}\log\norm{P_{\lambda^{(m,n)}}\varphi^{\otimes n}}^2\le I_{\epsilon_m}+\epsilon_m.
\end{equation}
For every $n\in\naturals$ let $M_n=\max\setbuild{m\in[n]}{n\ge n_0(m)}$. Then $(M_n)_{n\in\naturals}$ is an increasing sequence that is clearly unbounded. We claim that the sequence $\lambda^{(n)}:=\lambda^{(M_n,n)}$ satisfies the requirements. Indeed,
\begin{equation}
\norm[1]{\frac{1}{n}\lambda^{(n)}-\overline{\lambda}}=\norm[1]{\frac{1}{n}\lambda^{(M_n,n)}-\overline{\lambda}}\le\epsilon_{M_n}\to 0
\end{equation}
and $n\ge n_0(M_n)$, therefore
\begin{equation}
\begin{split}
\limsup_{n\to\infty}-\frac{1}{n}\log\norm{P_{\lambda^{(n)}}\varphi^{\otimes n}}^2
 & = \limsup_{n\to\infty}-\frac{1}{n}\log\norm{P_{\lambda^{(M_n,n)}}\varphi^{\otimes n}}^2  \\
 & \le \limsup_{n\to\infty}(I_{\epsilon_{M_n}}+\epsilon_{M_n}) = \ratefunction{\overline{\lambda}}{\varphi}.
\end{split}
\end{equation}
The opposite bound on the limit inferior follows from Lemma~\ref{lem:subsequenceliminfbound}.
\end{proof}

\bibliography{refs}{}

\end{document}